\documentclass[english]{article}
\usepackage[table]{xcolor}

\usepackage[T1]{fontenc}
\usepackage[latin9]{inputenc}
\usepackage{amsmath,amsthm} 
\usepackage{graphicx}
\usepackage{amssymb}
\usepackage{enumitem}
\usepackage{natbib} 
\usepackage{bussproofs}
\usepackage{fullpage}
\usepackage{thmtools,thm-restate}
\usepackage{tikz}
\usepackage{pgf} 
\usetikzlibrary{patterns,automata,arrows,shapes,snakes,topaths,trees,backgrounds,positioning,through,calc}
\usepackage{setspace}
\usepackage{multicol}
\usepackage{graphicx}
\usepackage{stmaryrd}  
\usepackage{hyperref}
\hypersetup{colorlinks=true, citecolor=blue, linkcolor=blue}  

\usepackage{array, makecell}

\usepackage{natbib}
\setcitestyle{round}

\definecolor{medgreen}{rgb}{0.0, 0.75, 0.0}

\definecolor{darkgreen}{rgb}{0.0, 0.4, 0.0}
\usepackage{changes}
\colorlet{Changes@Color}{darkgreen}

\usepackage{doi}

\theoremstyle{definition}
\newtheorem{theorem}{Theorem}[section]

\newtheorem{proposition}[theorem]{Proposition}
\newtheorem{corollary}[theorem]{Corollary}
\newtheorem{example}[theorem]{Example}
\newtheorem{definition}[theorem]{Definition}
\newtheorem{lemma}[theorem]{Lemma}
  
\newtheorem{remark}[theorem]{Remark}

\usepackage[page,toc,titletoc,title]{appendix}

\usepackage{datetime}


\onehalfspace

\begin{document} 

\settimeformat{ampmtime}

 \title{Axioms for Defeat in Democratic Elections}
 \author{Wesley H. Holliday$^\dagger$ and Eric Pacuit$^\ddagger$ \\ \\
 $\dagger$ University of California, Berkeley {\normalsize (\href{mailto:wesholliday@berkeley.edu}{wesholliday@berkeley.edu})} \\
 $\ddagger$ University of Maryland {\normalsize (\href{mailto:epacuit@umd.edu}{epacuit@umd.edu})}}
 
\date{{\normalsize Published in \textit{Journal of Theoretical Politics}, Vol.~33(4), 475-524, 2021}.}
 
\maketitle

 \begin{abstract}
 We propose six axioms concerning when one candidate should defeat another in a democratic election involving two or more candidates. Five of the axioms are widely satisfied by known voting procedures. The sixth axiom is a weakening of Kenneth Arrow's famous condition of the Independence of Irrelevant Alternatives (IIA). We call this weakening Coherent IIA. We prove that the five axioms plus Coherent IIA single out a method of determining defeats studied in our recent work: Split Cycle. In particular, Split Cycle provides the most resolute definition of defeat among any satisfying the six axioms for democratic defeat. In addition, we analyze how Split Cycle escapes Arrow's Impossibility Theorem and related impossibility results.
 \end{abstract}
 
 \tableofcontents

\section{Introduction}\label{IntroDefeat}

In the abstract for his lecture at a 2017 Lindau Nobel Laureate Meeting, Eric Maskin \citeyearpar{Maskin2017} claimed that ``The systems that most countries use to elect presidents are deeply flawed,'' a claim defended in writing by Maskin and Sen \citeyearpar{Maskin2016,Maskin2017a,Maskin2017b}. In fact, the issue goes far beyond presidential elections: the same voting systems are used in elections ranging from national elections to elections in small committees and clubs. In our view, the key issue  highlighted by Maskin and Sen can be stated in terms of the following normative principle (closely related to what voting theorists call Condorcet consistency, defined below\footnote{For the relation between Majority Defeat and Condorcet consistency, see Remark \ref{MajCon}.}).
\begin{itemize}
\item[] Majority Defeat: if a candidate loses an election (before any tiebreaking), they must have been \textit{defeated} by some other candidate in the election, and a candidate should defeat another only if a \textit{majority} of voters prefer the first candidate to the second.\footnote{This principle is not intended to apply to a candidate who ties for the win but ultimately loses a tiebreaker to another candidate. Furthermore, the  principle is intended only for elections in which the outcome of the election is determined solely by the votes cast in that election, as in conventional elections. The principle is not intended as a constraint on voting systems in which the outcome of an election is determined in part by votes cast in previous elections and previous election outcomes (see, e.g., \citealt{Harrenstein2020}).}
\end{itemize}

\noindent As is well known, widely used voting systems can violate the principle of Majority Defeat.

\begin{example}\label{BushGore} In the 2000 U.S. presidential election in Florida,  George W. Bush defeated Al Gore and Ralph Nader according to Plurality voting, which only allows voters to vote for one candidate. Yet based on the plausible inference that most Nader voters preferred Gore to Bush (see \citealt{Magee2003}), it follows that a majority of all voters preferred Gore to Bush.
\end{example}

\begin{example}\label{Burlington} In the 2009 mayoral election in Burlington, Vermont, the Progressive candidate Bob Kiss defeated the Democratic candidate Andy Montroll according to Instant Runoff Voting (defined in Example \ref{OtherVCCRs} below), but Montroll was preferred to each of the other candidates including Kiss by majorities of voters, according to the ranked ballots collected.
\end{example}

\begin{example}\label{Trump} During the 2016  U.S. presidential primary season, an NBC News/Wall Street Journal poll (March 3-6) asked respondents both for their top choice and their preference between Donald Trump and each of Ted Cruz, John Kasich, and Marco Rubio. Trump was the Plurality winner, receiving 30\% of first place votes, but Cruz, Kasich, and Rubio were each preferred to Trump by 57\%, 57\%, and 56\% of respondents, respectively (see \citealt{Kurrild-Klitgaard2018} concerning statistical significance). For further discussion of whether another Republican might have been majority preferred to Trump, see \citealt{Maskin2016}, \citealt{Maskin2017}, \citealt{Kurrild-Klitgaard2018}, and \citealt{Woon2020}.
\end{example}

\noindent For related examples outside the U.S., see, e.g., \citealt[\S~20.3.2]{Kaminski2015} and \citealt{Feizi2020}. 

The above failures of Majority Defeat involve \textit{spoiler effects}. In Example \ref{BushGore}, although it is likely that a majority of voters preferred Gore to Bush and also preferred Gore to Nader, Nader's inclusion in the race spoiled the election for Gore, handing victory  to Bush. In Example \ref{Burlington}, although a majority of voters preferred Montroll to Kiss and a majority preferred Montroll to the Republican candidate, Kurt Wright, Wright's inclusion in the race spoiled the election for Montroll, handing victory to Kiss. Finally, in Example \ref{Trump}, although the NBC News/Wall Street Journal poll did not ask for respondents' preferences between Cruz, Kasich, and Rubio, if one of them was majority preferred to the other two, then it would be reasonable to call the latter two spoilers for the first, as their inclusion in the poll handed the plurality victory to Trump.

What other benefits might respecting Majority Defeat provide besides mitigating spoiler effects of the kind shown above? Maskin and Sen \citeyearpar{Maskin2017a} make the following conjecture:
\begin{quote}
[M]ajority rule may reduce polarization. A centrist like Bloomberg [in the 2016 U.S. presidential election] may not be ranked first by a
large proportion of voters [and hence cannot win under Plurality], but can still be elected [with the backing of majorities against each other candidate] if viewed as a good compromise. Majority rule also encourages public debate about a larger group of potential candidates [since more candidates can participate without worry of their being spoilers], bringing us closer to John Stuart Mill's ideal of democracy as ``government by discussion.''
\end{quote}

To respect the principle of Majority Defeat, we need to collect ballots in which voters rank the candidates in the election.\footnote{One could collect even more information from each voter than a rankings of the candidates: e.g., a ranking plus a distinguished set of ``approved'' candidates (cf.~\citealt{Brams2009}) or a grading of each candidate (cf.~\citealt{Balinski2010}) from which a ranking can be derived. In this paper, we assume that only rankings of the candidates are collected from voters.} Due to the possibility of strategic voting, we cannot guarantee that voters' rankings of the candidates always reflect their sincere preferences (see \citealt{Taylor2005}), but one can try to choose voting procedures that provide fewer incentives for strategic voting (see, e.g., \citealt{Chamberlin1985}, \citealt{Nitzan1985}, \citealt{Bassi2015}, \citealt{HP2019}). Assuming we collect ranked ballots, a wide variety of voting procedures become available (see, e.g., \citealt{Brams2002}, \citealt{Pacuit2019}, and Examples \ref{OtherVCCRs} and \ref{BeatPathEx}~below).

One obvious idea for satisfying the principle of Majority Defeat is to say that one candidate defeats another \textit{if and only if} a majority of voters prefer the first candidate to the second. Notoriously, however, this can result in every candidate being defeated, leaving no candidate who wins. In particular, there can be a \textit{majority cycle}, wherein a majority of voters prefer $a$ to $b$, a majority of voters prefer $b$ to $c$, and a majority of voters prefer $c$ to $a$ (\citealt{Condorcet1785}). Majority cycles may also involve more than three candidates. This so-called Paradox of Voting is perhaps the main theoretical obstacle to the possibility of rational democratic decision making with more than two candidates. Riker \citeyearpar{Riker1982} has famously argued that the Paradox of Voting, along with the related Arrow Impossibility Theorem (\citealt{Arrow1963}), destroys the notion of a coherent ``will of the people'' in a democracy.\footnote{Wolf \citeyearpar{Wolf1970} takes these results to show that ``majority rule is fatally flawed by an internal inconsistency'' (p.~59), and Hardin \citeyearpar{Hardin1990} takes them to cast ``doubt on the conceptual coherence of majoritarian democracy'' (p.~184). For further discussion, see \citealt{Risse2001,Risse2009}.} Although there is not yet enough empirical research to know how prevalent majority cycles are in real elections of various scales,\footnote{As Van Deemen \citeyearpar{VanDeemen2014} remarks, ``it is remarkable to see that the empirical research on the paradox has been conducted mainly for large elections. Collective decision making processes in relatively small committees, such as corporate boards of directors, management teams in organizations, government cabinets, councils of political parties and so on, have hardly been studied'' (p.~325) (though see \citealt{Mattei2013}). Concerning the relevance of such empirical research, Ingham \citeyearpar{Ingham2019} argues that ``Arrow's theorem and related results threaten the populist's principle of democratic legitimacy even if majority preference cycles never occur'' (p.~97).} majority cycles have been found in some large elections (see \citealt{VanDeemen2014}). In a typical election, we expect (or at least hope) that there will be a \textit{Condorcet winner}---a candidate $a$ such that for every candidate $b$, more voters prefer $a$ to $b$ than $b$ to $a$---in which case some voting theorists believe that the choice is clear: elect the Condorcet winner (see, e.g., \citealt{Felsenthal1992,Maskin2017a,Maskin2017b}). A voting method is \textit{Condorcet consistent} if it chooses as the unique winner of an election the Condorcet winner, whenever a Condorcet winner exists. But there must be some backup plan in place in case there is no Condorcet winner.

In the absence of a Condorcet winner, Maskin and Sen \citeyearpar{Maskin2017a} suggest ``having a runoff between the two top candidates,'' but defining ``the two top candidates'' faces some of the same difficulties as defining ``the best candidate.'' In a recent paper (\citealt{HP2020}), we study a voting procedure that we call Split Cycle, which provides a different backup plan for the case where no Condorcet winner exists.\footnote{After posting \citealt{HP2020} online, we learned from Jobst Heitzig of his notion of the ``immune set'' discussed in a 2004 post on the Election-Methods mailing list (\citealt{Heitzig2004}), which is equivalent to the set of winners for Split Cycle after replacing `stronger' with `at least as strong' in Heitzig's definition in the post. After submitting the present paper, we learned from Markus Schulze of Steve Eppley's notion of the ``Beatpath Criterion Method'' in a 2000 post on the Election-Methods mailing list (\citealt{Eppley2000}), which is equivalent to a version of Split Cycle that measures strength of majority preference using winning votes (the number of voters who rank $x$ above $y$), though \citealt{HP2020} focuses on the version that uses margin of victory (the number of voters who rank $x$ above $y$ minus the number of voters who rank $y$ above $x$).} Instead of saying \textit{if there is a Condorcet winner, elect that person, and if not, do something else with a different justification}, Split Cycle provides a unified rule for cases with or without Condorcet winners:
\begin{itemize}
\item[] In an election with candidates $x$ and $y$, say that $x$ \textit{wins by a margin of $n$ over $y$} when there are $n$ more voters who prefer $x$ to $y$ than who prefer $y$ to $x$.  Then $x$ defeats $y$ according to Split Cycle if $x$ wins by more than $n$ over $y$ for the smallest number $n$ such that there is no majority cycle containing $x$ and $y$ in which each candidate wins by more than $n$ over the next candidate in the cycle.\end{itemize}

\noindent While there may be no Condorcet winner, there is always an undefeated candidate according to Split Cycle (if there is more than one, a tiebreaking process must be applied---cf.~Remark \ref{Tiebreaking}). An intuitive way\footnote{See \citealt{HP2020} for a more efficient algorithm for computing the undefeated candidates.} to determine the Split Cycle defeat relation is as follows: 
\begin{enumerate}
\item In each majority cycle, identify the wins with the smallest margin in that cycle. 
\item After completing step 1 for all cycles, discard the identified wins. All remaining wins count as defeats.
\end{enumerate}
\noindent As we show, Split Cycle mitigates spoiler effects (see Section \ref{CoherentIIASection}) and has several other virtues, including avoiding the so-called Strong No Show Paradox (see \citealt{HP2020,HP2021PI}).

In this paper, we arrive at Split Cycle (defined formally in Section \ref{SCsection}) by another route. We propose six general axioms concerning when one candidate should defeat another in a democratic election involving two or more candidates (Section \ref{AxiomsSection}). Five of the axioms are widely satisfied by known voting procedures. The sixth axiom is a weakening of Kenneth Arrow's famous condition of the Independence of Irrelevant Alternatives (IIA) (\citealt{Arrow1963}). We call this weakening Coherent IIA. Arrow's IIA states that if in two elections, all voters rank candidate $x$ vs. candidate $y$ in the same way, then if $x$ defeats $y$ in the first election, $x$ must also defeat $y$ in the second election. Coherent IIA agrees provided that the second election does not involve greater \textit{incoherence} with respect to $x$ and $y$, in the sense of new majority cycles or stronger majority cycles involving $x$ and $y$. For if the second election involves greater incoherence, we may need to suspend the judgment that $x$ defeats $y$ that we could coherently accept in the first election. Both by itself and together with other axioms, Coherent IIA has a number of desirable consequences. Already by itself Coherent IIA rules out---as IIA does---a flip from $x$ defeating $y$ to $y$ defeating $x$ in two elections in which all voters rank $x$ vs. $y$ in the same way; and together with other natural axioms, Coherent IIA implies Majority Defeat.

The first half of the paper culminates in a proof that the five axioms plus Coherent IIA single out Split Cycle (Section \ref{CharSection}). In particular, our main result is that Split Cycle provides the most resolute definition of defeat of any satisfying the six axioms for democratic defeat. In the second half, we analyze  how Split Cycle manages to escapes Arrow's Impossibility Theorem and related impossibility results in social choice theory (Section \ref{EscapeSection}). The answer is twofold: we weaken IIA to Coherent IIA, and we relax Arrow's assumptions about the properties of the defeat relation between candidates. We explain how neither of these moves is sufficient by itself to escape Arrow-like impossibility theorems. But by doing both, Split Cycle provides a compelling response, we think, to both the Paradox of Voting and Arrow's Impossibility Theorem.

A key aspect of our characterization of Split Cycle using the six axioms for defeat is that we work with a model in which elections can have different sets of voters and different sets of candidates, just as they do in reality. Given the importance of this variable-election setting to our characterization, we consider how standard impossibility results for a fixed set of candidates/voters can be adapted to and even strengthened in the variable-election setting, and yet how Split Cycle still escapes them (Sections \ref{IIAVCCR} and \ref{AlphaVoting}). One of the methodological lessons of the paper, in our view, is the value of working in a variable-election framework.

We start in Section \ref{Prelim} by reviewing the formal framework we will use to conduct our analysis.

\section{Voting methods and collective choice rules}\label{Prelim}

As suggested in Section \ref{IntroDefeat}, we work in a \textit{variable-voter} and \textit{variable-candidate} setting. This means that our group decision method can input  elections---formalized as \textit{profiles} below---with different sets of voters and different sets of candidates (see Remark \ref{Subtleties}.\ref{Subtleties1} for a comparison with a fixed-voter and fixed-candidate setting). To allow sets of voters and candidates of arbitrary (but finite) size in elections, we first fix infinite sets $\mathcal{V}$ and $\mathcal{X}$ of \textit{voters} and \textit{candidates}, respectively. A given election will use only finite subsets $V\subset \mathcal{V}$ and $X\subset \mathcal{X}$. We consider elections in which each voter in the election submits a ranking of the candidates in the election. For simplicity, in this paper we assume that each voter submits a \textit{strict linear order} on the set $X$ of candidates, i.e., a binary relation $P$ on $X$ satisfying the following conditions for all $x,y,z\in X$:
\begin{itemize}
\item asymmetry: if $xPy$ then not $yPx$;
\item transitivity: if $xPy$ and $yPz$, then $xPz$;
\item connectedness: if $x\neq y$, then $xPy$ or $yPx$.
\end{itemize}
We take $xPy$ to mean that the voter strictly prefers candidate $x$  to candidate $y$. In practice, one may wish to allow voters not to rank all the candidates or even to indicate indifference between candidates. The voting procedures discussed below can be generalized to this setting, as we will discuss for our favored procedure in future work, but doing so raises some choice points that are not essential to the main ideas in this paper.

We formalize the notion of an election as a function associating with each voter their ranking of the candidates. For a set $X$, let $\mathcal{L}(X)$ be the set of all strict linear orders on $X$.

 \begin{definition}\label{ProfileDef}
 A \textit{profile} is a function $\mathbf{P}: V\to \mathcal{L}(X)$ for some nonempty finite $V\subset \mathcal{V}$ and nonempty finite $X\subset \mathcal{X}$, which we denote by $V(\mathbf{P})$ (called the set of \textit{voters in $\mathbf{P}$})  and $X(\mathbf{P})$ (called the set of \textit{candidates in $\mathbf{P}$}), respectively. We call $\mathbf{P}(i)$ voter $i$'s \textit{ballot}, and we write `$x\mathbf{P}_iy$' for $(x,y)\in\mathbf{P}(i)$.
\end{definition}

When we display profiles, we show their ``anonymized form'' that records only the number of candidates with each type of ballot, rather than the  identities of the voters. For example:
\begin{center}
\begin{minipage}{2in}\begin{tabular}{ccc}
$4$ & $2$ & $3$   \\\hline
$a$ & $b$ &  $c$ \\
$b$ &  $c$ & $a$ \\
$c$ &  $a$ &  $b$ \\
\end{tabular}\end{minipage}
\end{center}
The above diagram indicates that two voters rank $b$ above $c$ above $a$ (notation: $bca$), etc.

It will be important later to consider the restriction of a profile to a subset $Y$ of the candidates: we erase from each voter's ballot any candidates not in $Y$, leaving the ranking of the candidates in $Y$ unchanged.

\begin{definition}\label{Restriction}
Given a binary relation $P$ on $X$ and $Y\subseteq X$, let $P_{\mid Y}$ be the restriction of $P$ to the set $Y$, i.e., $P_{\mid Y}=P\cap (Y\times Y)$. Given a profile $\mathbf{P}$, let $\mathbf{P}_{\mid Y}$ be the profile with $X(\mathbf{P}_{\mid Y})= Y$ and $V(\mathbf{P}_{\mid Y})=V(\mathbf{P})$ obtained from $\mathbf{P}$ by restricting each voter's ballot to the set $Y$.\footnote{I.e., for all $i\in V(\mathbf{P})$, $\mathbf{P}_{\mid Y}(i)=\mathbf{P}(i)_{\mid Y}$.}
\end{definition}

We now consider two different kinds of group decision methods, differing in what they output. The first kind outputs a set of potential winners for the election.

\begin{definition}\label{VotingMethod} A \textit{voting method} is a function $F$ on the domain of all profiles that returns for any profile $\mathbf{P}$ a nonempty subset $F(\mathbf{P})$ of the candidates in $\mathbf{P}$, i.e., $\varnothing\neq F(\mathbf{P})\subseteq X(\mathbf{P})$.
\end{definition}
\noindent As usual, if $F(\mathbf{P})$ contains multiple candidates, we assume that some further tiebreaking process would then apply, though we do not fix the nature of this process (see \citealt[pp.~14-15]{Schwartz1986} for further discussion). If $x\not\in F(\mathbf{P})$, this means that $x$ is excluded from the rest of the process that leads to the ultimate winner.

The second kind of group decision method outputs an asymmetric binary relation on the set of candidates. In social choice theory, this relation is typically called the ``strict social preference'' relation. We interpret this binary relation as a \textit{defeat relation} for the election in the sense of Section \ref{IntroDefeat}.

\begin{definition} A \textit{variable-election collective choice rule} (VCCR) is a function $f$ on the domain of all profiles such that for any profile $\mathbf{P}$, $f(\mathbf{P})$ is an asymmetric binary relation on $X(\mathbf{P})$, which we call \textit{the defeat relation for $\mathbf{P}$ according to $f$}. For $x,y\in X(\mathbf{P})$, we say that \textit{$x$ defeats $y$ in $\mathbf{P}$ according to $f$} when $(x,y)\in f(\mathbf{P})$.
\end{definition}

A well-known special case of a collective choice rule is what Arrow called a \textit{social welfare function} (SWF). The output of an SWF is a \textit{strict weak order}, i.e., a binary relation $P$ on $X$ satisfying asymmetry and the condition that for all $x,y,z\in X:$
\begin{itemize}
\item if $xPy$, then $xPz$ or $zPy$.
\end{itemize}
Note that these conditions imply that $P$ is transitive. In the variable-election setting, we define the following.

\begin{definition} A \textit{variable-election social welfare function} (VSWF) is a VCCR $f$ such that for any profile $\mathbf{P}$, $f(\mathbf{P})$ is a strict weak order.
\end{definition}

For readers familiar with the standard setup in social choice theory, we note some subtleties about our definitions.

\begin{remark}\label{Subtleties}$\,$
\begin{enumerate}
\item\label{Subtleties1} We add the modifier  `variable-election' because the term `collective choice rule' due to Sen \citeyearpar[Ch.~2$^*$]{Sen2017} appears in a fixed-voter and fixed-candidate setting.  In this setting, one begins by fixing a nonempty set $V$ of voters and a nonempty set $X$ of candidates; one then defines a collective choice rule (CCR) as a function that takes as its input a profile for $V$ and $X$ (see Section \ref{FixedImposs}). One could call this a ``variable election'' setting insofar as the sizes of $V$ and $X$ are not specified and definitions of CCRs and axioms concerning CCRs do not refer to specific numbers of voters and candidates. However, it is not a variable election setting in our sense, since the domain of a CCR cannot contain both a profile whose set of voters is $\{i,j,k\}$ and a profile whose set of voters is $\{i,j\}$; likewise, it cannot contain both a profile whose set of candidates is $\{a,b,c\}$ and a profile whose set of candidates is $\{a,b\}$, etc. Yet there are important axioms in voting theory concerning the addition or removal of voters or candidates: not only some of the axioms proposed below, but also, for instance, axioms concerning adding voters who support a given candidate (see, e.g., \citealt{Felsenthal2016}) or adding a candidate who is a ``clone'' of another candidate (\citealt{Tideman1987}). These axioms cannot be formalized in terms of a CCR whose domain contains only profiles for $V$ and $X$.

\item In social choice theory, one often defines the output of a CCR (as Sen does) to be a ``weak social preference'' relation $R$ that is reflexive instead of asymmetric. For social welfare functions, the choice does not matter, because strict weak orders $P$ are in one-to-one correspondence with complete and transitive relations $R$.\footnote{A relation $R$ on $X$ is complete if for all $x,y\in X$, we have $xRy$ or $yRx$.} However, since we aim to study the concept of \textit{defeat}, an asymmetric relation, we have defined VCCRs accordingly.
\item For simplicity, we build the axiom of Universal Domain into the definition of a VCCR, but one could of course define a notion of VCCR where only certain profiles are in the domain of $f$ (cf.~\citealt{Gaertner2001}).
\end{enumerate}
\end{remark}

Any VSWF $f$ induces a voting method $\overline{f}$ such that for any profile $\mathbf{P}$, $\overline{f}(\mathbf{P})$ is the set of candidates who are not defeated by any candidates in $\mathbf{P}$ according to $f$. That $f(\mathbf{P})$ is a strict weak order implies that $\overline{f}(\mathbf{P})$ is nonempty, but in fact a much weaker condition is sufficient---namely, acyclicity. 

\begin{definition}\label{AcyclicDef} Let $P$ be an asymmetric binary relation on a set $X$. A \textit{cycle in $P$} is a sequence $x_1,\dots,x_n$ of elements of $X$ such that $x_1Px_2,\dots,x_{n-1}Px_n$, $x_n=x_1$, and all elements are distinct except $x_1$ and $x_n$.\footnote{In requiring that all elements are distinct except $x_1$ and $x_n$, we are using the term `cycle' for what is called a \textit{simple cycle}.} The relation $P$ is \textit{acyclic} if there is no cycle in $P$. A VCCR $f$ is \textit{acyclic} if for all profiles $\mathbf{P}$, $f(\mathbf{P})$ is acyclic.
\end{definition}

Any acyclic VCCR induces a voting method that outputs for a given profile the set of undefeated candidates. All defeated candidates are excluded from the rest of the process that leads to the ultimate winner.

\begin{lemma}\label{VCCRtoVoting} Given any acyclic VCCR $f$, the function $\overline{f}$  on the set of profiles defined by 
\[\overline{f}(\mathbf{P})=\{x\in X(\mathbf{P})\mid \mbox{there is no }y\in X(\mathbf{P}): y\mbox{ defeats }x\mbox{ in $\mathbf{P}$ according to }f\}\]
is a voting method, as $\varnothing\neq \overline{f}(\mathbf{P})\subseteq X(\mathbf{P})$.
\end{lemma}

Given a voting method $F$, we can consider the acyclic VCCRs from which $F$ arises as in Lemma \ref{VCCRtoVoting}.

\begin{definition}\label{DefeatRationalize} Let $F$ be a voting method and $f$ a VCCR. Then $F$ is \textit{defeat rationalized by $f$} if $F=\overline{f}$.
\end{definition}

Let us now review some standard VCCRs. Several of the VCCRs are based on the majority preference relation, defined as follows.

\begin{definition}Given a profile $\mathbf{P}$ and $x,y\in X(\mathbf{P})$, we say that $x$ is \textit{majority preferred to $y$ in $\mathbf{P}$} (and $y$ is \textit{majority dispreferred to $x$ in $\mathbf{P}$}) if more voters rank $x$ above $y$ in $\mathbf{P}$ than rank $y$ above $x$ in $\mathbf{P}$. We write $x\to_\mathbf{P} y$ (or $x\to y$ if $\mathbf{P}$ is clear from context) to indicate that $x$ is majority preferred to $y$ in $\mathbf{P}$. 

A \textit{majority cycle in $\mathbf{P}$} is a cycle in the relation $\to_\mathbf{P}$.
\end{definition}

\begin{example}\label{OtherVCCRs} $\,$
\begin{enumerate}
\item \textbf{Simple Majority}. For $x,y\in X(\mathbf{P})$, $x$ defeats $y$ in $\mathbf{P}$ if and only if $x\to y$.
\item \textbf{Covering} (\citealt{Gillies1959,Fishburn1977,Miller1980}). For $x,y\in X(\mathbf{P})$, say that \textit{$x$ left-covers $y$ in $\mathbf{P}$} if for all $z\in X(\mathbf{P})$, if $z\to x$, then $z\to y$; and \textit{$x$ right-covers $y$ in $\mathbf{P}$} if for all $z\in X(\mathbf{P})$, if $y\to z$, then $x\to z$. (Left-covering and right-covering are equivalent if $\mathbf{P}$ has an odd number of voters but not for an even number of voters.)  We say that $x$ defeats $y$ in $\mathbf{P}$ according to the Left Covering VCCR (resp.~Right Covering VCCR) if $x\to y$ and $x$ left-covers $y$ (resp.~$x\to y$ and $x$ right-covers $y$). We say that $x$ defeats $y$ in $\mathbf{P}$ according to the Fishburn VCCR if $x$ left-covers $y$ but $y$ does not left-cover $x$.
\item \textbf{Copeland} (\citealt{Copeland1951}). The Copeland score of a candidate $x$ in profile $\mathbf{P}$ is the number of candidates to whom $x$ is majority preferred in $\mathbf{P}$ minus the number of candidates who are majority preferred to $x$ in $\mathbf{P}$: $|\{z\in X(\mathbf{P})\mid x\to z\} | - | \{z\in X(\mathbf{P})\mid z\to x\}|$. Then for $x,y\in X(\mathbf{P})$, $x$ defeats $y$ in $\mathbf{P}$ if and only if the Copeland score of $x$ is greater than the Copeland score of $y$.
\item \textbf{Borda}. The Borda score of a candidate $x$ in profile $\mathbf{P}$ is calculated as follows: for every voter who ranks $x$ in last place, $x$ receives 0 points, and for every voter who ranks $x$ in second to last place, $x$ receives 1 point, and so on. That is, for every voter who ranks $x$ in $k$ places above last place, $x$ receives $k$ points. The sum of the points that $x$ receives is $x$'s Borda score in $\mathbf{P}$. Then for $x,y\in X(\mathbf{P})$, $x$ defeats $y$ in $\mathbf{P}$ if and only if the Borda score of $x$ is greater than the Borda score of $y$.
\item \textbf{Plurality}. The plurality score of a candidate $x$ in profile $\mathbf{P}$ is the number of voters who rank $x$ in first place. Then for $x,y\in X(\mathbf{P})$, $x$ defeats $y$ in $\mathbf{P}$ if and only if the plurality score of $x$ is greater than the plurality score of $y$.
\item \textbf{Instant Runoff} (\textbf{Hare}). Given a profile $\mathbf{P}$, define a sequence $\mathbf{P}_0,\dots,\mathbf{P}_n$ of profiles as follows. First, $\mathbf{P}_0=\mathbf{P}$. Second, given a profile $\mathbf{P}_k$ in the sequence, if all candidates in $\mathbf{P}_k$ have the same plurality score, set $n=k$ to end the sequence; otherwise, where $A_k$ is the set of candidates whose plurality score in $\mathbf{P}_k$ is above the lowest plurality score of a candidate in $\mathbf{P}_k$, let $\mathbf{P}_{k+1}$ be obtained from $\mathbf{P}_k$ by restricting the set of candidates to $A_k$, i.e., $\mathbf{P}_{k+1}= (\mathbf{P}_k)_{\mid A_k}$.\footnote{\label{IRVPUT}When there is more than one candidate with lowest plurality score, this definition of Instant Runoff, taken from \citealt[p.~7]{Taylor2008}, eliminates all such candidates. For discussion of an alternative ``parallel universe'' approach to dealing with ties for the lowest plurality score, see \citealt[\S~3]{Freeman2015}.} The Hare score of candidate $x$ in $\mathbf{P}$ is  the number of rounds of elimination that $x$ survives, i.e., the greatest $k$ such that $x\in A_k$. Then $x$ defeats $y$ in $\mathbf{P}$ if and only if the Hare score of $x$ is greater than the Hare score of $y$. 
\end{enumerate}
VCCRs 2-6 are all acyclic---but 1 is not, due to the possibility of majority cycles---and 3-6 are VSWFs. For axiomatic characterizations of the Copeland and Borda VSWFs, see \citealt{Rubinstein1980} for Copeland and \citealt{Nitzan1981} and \citealt{Mihara2017} for Borda.
\end{example}

VCCRs 1-4 all have an important property in common: their output depends only on the \textit{margins} between candidates in the given profile.

\begin{definition}\label{MarginDef} Let $\mathbf{P}$ be a profile and $x,y\in X(\mathbf{P})$. The \textit{margin of $x$ over $y$ in $\mathbf{P}$} is the number of voters who rank $x$ above $y$ in $\mathbf{P}$ minus the number of voters who rank $y$ above $x$ in $\mathbf{P}$.\footnote{Note that the margin of $x$ over $y$ is \textit{negative} when $y$ is majority preferred to $x$.} Let $Margin_\mathbf{P}(x,y)$  be the margin of $x$ over $y$ in $\mathbf{P}$.

The \textit{margin graph of $\mathbf{P}$}, $\mathcal{M}(\mathbf{P})$, is the directed graph with weighted edges whose set of nodes is $X(\mathbf{P})$ with an edge from $x$ to $y$ when $x$ is majority preferred to $y$, weighted by the margin of $x$ over $y$ in $\mathbf{P}$.
\end{definition}

\begin{example} For a profile $\mathbf{P}$ shown in anonymized form on the left, its margin graph $\mathcal{M}(\mathbf{P})$ is shown on the right:
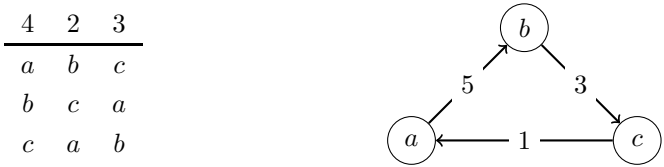
\begin{center}
\begin{minipage}{2in}\begin{tabular}{ccc}
$4$ & $2$ & $3$   \\\hline
$a$ & $b$ &  $c$ \\
$b$ &  $c$ & $a$ \\
$c$ &  $a$ &  $b$ \\
\end{tabular}\end{minipage}\begin{minipage}{2in}\begin{tikzpicture}

\node[circle,draw, minimum width=0.25in] at (0,0) (a) {$a$}; 
\node[circle,draw,minimum width=0.25in] at (3,0) (c) {$c$}; 
\node[circle,draw,minimum width=0.25in] at (1.5,1.5) (b) {$b$}; 

\path[->,draw,thick] (b) to node[fill=white] {$3$} (c);
\path[->,draw,thick] (c) to node[fill=white] {$1$} (a);
\path[->,draw,thick] (a) to node[fill=white] {$5$} (b);

\end{tikzpicture}
\end{minipage}\end{center}
\end{example}

\noindent Clearly the edge relation in $\mathcal{M}(\mathbf{P})$ must be \textit{asymmetric}, since if $x$ is majority preferred to $y$, then $y$ is not majority preferred to $x$. Also note that if there is an even number of voters, then it may be that neither $x$ nor $y$ is majority preferred to the other, in which case there is no edge from $x$ to $y$ or from $y$ to $x$ in $\mathcal{M}(\mathbf{P})$. Thus, the underlying graph of $\mathcal{M}(\mathbf{P})$ is not necessarily a \textit{tournament}, which is a directed graph whose edge relation is asymmetric and \textit{connected}, i.e., if $x\neq y$, then there is an edge from $x$ to $y$ or an edge from $y$ to $x$. However, when the number of voters is odd, then the directed graph is a tournament.

Now the idea that the output of a VCCR depends only on margins can be formalized as follows.\footnote{\label{Fish1}Cf.~De Donder et al.'s \citeyearpar{DeDonder2000} notion of C1.5 functions. Note that Fishburn's \citeyearpar{Fishburn1977} C2 functions can use not only the \textit{difference} between the number of voters who prefer $x$ to $y$ and the number of voters who prefer $y$ to $x$ but also those two numbers themselves. The Pareto VCCR in Example \ref{ParetoReversal} is C2 but not margin based (not C1.5).}

\begin{definition}\label{MarginBasedDef} A VCCR $f$ is \textit{margin based} if for any profiles $\mathbf{P}$ and $\mathbf{P}'$, if $\mathcal{M}(\mathbf{P})=\mathcal{M}(\mathbf{P}')$, then $f(\mathbf{P})=f(\mathbf{P}')$.
\end{definition}

It is obvious that VCCRs 1-3 in Example \ref{OtherVCCRs} are margin based, but this is less obvious for Borda.

\begin{lemma}\label{BordaMargins} For any profile $\mathbf{P}$ and $x,y\in X(\mathbf{P})$, $x$ defeats $y$ according to the Borda VCCR if and only if the sum of the margins of $x$ over other all other candidates is greater than the sum of the margins of $y$ over all other candidates.\footnote{Remember that the margin of $x$ over $z$ is negative when $z$ is majority preferred to $x$.}
\end{lemma}

Other examples of margin-based VCCRs include the following.

\begin{example}\label{BeatPathEx} $\,$
\begin{enumerate}
\item \textbf{Weighted Covering} (\citealt{Dutta1999}, \citealt{Fernandez2018}). Given a profile $\mathbf{P}$ and $x,y\in X(\mathbf{P})$, $x$ defeats $y$ in $\mathbf{P}$ if $x\to y$ and for all $z\in X(\mathbf{P})$, $Margin_\mathbf{P}(x,z)\geq Margin_\mathbf{P}(y,z)$ (or equivalently $Margin_\mathbf{P}(z,x)\leq Margin_\mathbf{P}(z,y)$).
\item \textbf{Beat Path} (\citealt{Schulze2011}). Given a profile $\mathbf{P}$ and $x,y\in X(\mathbf{P})$ a \textit{path from $x$ to $y$} is a sequence $z_1,\dots,z_n$ of candidates with $z_1=x$ and $z_n=y$ such that each candidate is majority preferred to the next candidate in the sequence. The \textit{strength} of a path is the smallest margin between consecutive candidates in the path. Then $x$ defeats $y$ in $\mathbf{P}$ according to the Beat Path VCCR if the strength of the strongest path from $x$ to $y$ is greater than the strength of the strongest path from $y$ to $x$.
\end{enumerate}
\end{example}

\begin{remark}\label{QualMarg} Within the family of margin-based VCCRs, we can make a useful three-way distinction.
\begin{enumerate}
\item A \textit{majority graph} is any directed graph $M$ whose edge relation is asymmetric. Given a profile $\mathbf{P}$, the \textit{majority graph of $\mathbf{P}$}, $M(\mathbf{P})$, is the directed graph whose set of nodes is $X(\mathbf{P})$ with an edge from $x$ to $y$ when $x$ is majority preferred to $y$ in $\mathbf{P}$. We say that a VCCR $f$ is \textit{majority based} if for any profiles $\mathbf{P}$ and $\mathbf{P}'$, if $M(\mathbf{P})=M(\mathbf{P}')$, then $f(\mathbf{P})=f(\mathbf{P}')$.\footnote{\label{Fish2}Cf.~Fishburn's \citeyearpar{Fishburn1977} C1 functions.} 
The Simple Majority, Covering, and Copeland VCCRs in Example \ref{OtherVCCRs} are majority based in this sense. Similarly one can define the class of majority based voting methods (as distinguished from VCCRs), which generalize \textit{tournament solutions} (\citealt{Laslier1997}, \citealt{Brandt2016}) from tournaments to majority graphs.

\item A \textit{qualitative margin graph} is a pair $\mathbb{M}=(M,\prec)$ where $M$ is a majority graph and $\prec$ is a strict weak order on the set of edges of $M$. The \textit{qualitative margin graph of $\mathbf{P}$} is the pair $\mathbb{M}(\mathbf{P})=(M(\mathbf{P}), \prec_\mathbf{P})$ such that for any edges $(a,b)$ and $(c,d)$ in $M(\mathbf{P})$, we have $(a,b)\prec_\mathbf{P}(c,d)$ if $Margin_\mathbf{P}(a,b)<Margin_\mathbf{P}(c,d)$. We say that a VCCR $f$ is \textit{qualitative-margin based} if for any profiles $\mathbf{P}$ and $\mathbf{P}'$, if $\mathbb{M}(\mathbf{P})=\mathbb{M}(\mathbf{P}')$, then $f(\mathbf{P})=f(\mathbf{P}')$.\footnote{In terms of the C1, C1.5, and C2 classifications in Footnotes \ref{Fish1} and \ref{Fish2},  methods that are qualitative-margin based could be called~C1.25. Another example of a VCCR that is qualitative-margin based is the following, which defeat rationalizes the Simpson-Kramer Minimax method (\citealt{Simpson1969}, \citealt{Kramer1977}): $x$ defeats $y$ if $x$'s largest majority loss is smaller than $y$'s largest majority loss.} The Weighted Covering and Beat Path VCCRs in Example \ref{BeatPathEx} are qualitative-margin based in this sense, as is the Split Cycle VCCR defined in Section \ref{SCsection}, but none of these VCCRs are majority based.

\item\label{QualMarg3} A \textit{margin graph} is a weighted directed graph $\mathcal{M}$ such that: the edge relation of the graph is asymmetric; either all weights of edges are even positive integers or all weights of edges are odd positive integers; and if there are two nodes with no edge between them, then all weights are even. We already defined for a profile $\mathbf{P}$ the \textit{margin graph of $\mathbf{P}$} in Definition \ref{MarginBasedDef}, as well as \textit{margin-based} VCCRs in Definition \ref{MarginBasedDef}. Note that the Borda VCCR in Example \ref{OtherVCCRs} is margin based but not qualitative-margin based.
\end{enumerate}
\end{remark}

Finally, a useful fact for the study of margin-based VCCRs is that any abstract margin graph as in Remark \ref{QualMarg}.\ref{QualMarg3} can be realized as the margin graph of a profile.

\begin{theorem}[\citealt{Debord1987}] For any margin graph $\mathcal{M}$, there is a profile $\mathbf{P}$ such that $\mathcal{M}$ is the margin graph of $\mathbf{P}$.
\end{theorem}

\noindent In light of Debord's Theorem, one can construct margin graphs at will, without deriving them from particular profiles, when experimenting with the operation of margin-based VCCRs.

This concludes our review of basic notions. In the next section we turn to our preferred VCCR.

\section{Split Cycle}\label{SCsection}

In \citealt{HP2020}, we studied a voting method that we call Split Cycle. Here we formulate the Split Cycle VCCR that defeat rationalizes the Split Cycle voting method (recall Definition \ref{DefeatRationalize}). We give two formulations in Definition \ref{WinByDef} and Lemma \ref{SplittingLem}, respectively. The first definition of Split Cycle formalizes the definition given in Section \ref{IntroDefeat}. For a profile $\mathbf{P}$, candidates $x,y\in X(\mathbf{P})$, and natural number $n$, say that $x$ \textit{wins by $n$ over $y$} if the margin of $x$ over $y$ in $\mathbf{P}$ is $n$ (recall Definition \ref{MarginDef}).

\begin{definition}\label{WinByDef}
Given a profile $\mathbf{P}$ and candidates $x,y\in X(\mathbf{P})$, $x$ defeats $y$ in $\mathbf{P}$ according to Split Cycle if $x$ wins by more than $n$ over $y$ for the smallest number $n$ such that there is no majority cycle containing $x$ and $y$ in which each candidate wins by more than $n$ over the next candidate in the cycle.\end{definition}

The basic idea is that when the electorate's majority preference relation is \textit{incoherent}, in the sense that there is a majority cycle, this raises the threshold required for one candidate $a$ in the cycle to defeat another $b$---but not infinitely. If we raise the threshold $n$ sufficiently, then there will be no incoherence involving $a$ and $b$ with respect to the higher threshold, i.e., no cycles in the \textit{win by more than $n$} relation that contain $a$ and $b$. If the margin of $a$ over $b$ is greater than this sufficiently large $n$,  Split Cycle says that $a$ defeats $b$. Note, crucially, that for some other pair of candidates $a'$ and $b'$, the threshold for $a'$ to defeat $b'$ may be different, if $a'$ and $b'$ are contained in different majority cycles than $a$ and $b$ are.

\begin{example}\label{WinByEx} Consider a profile $\mathbf{P}$ with the following margin graph:
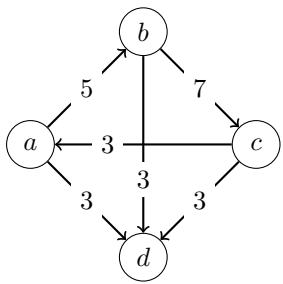
\begin{center}
\begin{minipage}{2in}\begin{tikzpicture}

\node[circle,draw, minimum width=0.25in] at (0,0) (a) {$a$}; 
\node[circle,draw,minimum width=0.25in] at (3,0) (c) {$c$}; 
\node[circle,draw,minimum width=0.25in] at (1.5,1.5) (b) {$b$}; 

\node[circle,draw,minimum width=0.25in] at (1.5,-1.5) (d) {$d$}; 

\path[->,draw,thick] (a) to node[fill=white] {$5$} (b);
\path[->,draw,thick] (b) to node[fill=white] {$7$} (c);
\path[->,draw,thick] (c) to[pos=.7] node[fill=white] {$3$} (a);

\path[->,draw,thick] (c) to node[fill=white] {$3$} (d);
\path[->,draw,thick] (b) to[pos=.7] node[fill=white] {$3$} (d);
\path[->,draw,thick] (a) to node[fill=white] {$3$} (d);

\end{tikzpicture}
\end{minipage}\end{center}
The only majority cycle is $a\to b\to c\to a$. Note that each candidate wins by \textit{more than $2$} over the next candidate in the sequence. However, it is not the case that each candidate wins by \textit{more than $3$} over the next candidate in the sequence. Thus, a threshold of  \textit{win by more than $3$} splits the $a\to b\to c\to a$ cycle:
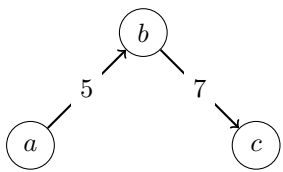
\begin{center}
\begin{minipage}{2in}\begin{tikzpicture}

\node[circle,draw, minimum width=0.25in] at (0,0) (a) {$a$}; 
\node[circle,draw,minimum width=0.25in] at (3,0) (c) {$c$}; 
\node[circle,draw,minimum width=0.25in] at (1.5,1.5) (b) {$b$}; 

\path[->,draw,thick] (a) to node[fill=white] {$5$} (b);
\path[->,draw,thick] (b) to node[fill=white] {$7$} (c);

\end{tikzpicture}
\end{minipage}\end{center}
Hence there is no incoherence involving $a$ and $b$ with respect to the \textit{win by more than $3$} relation. Then since $a$ wins by more than $3$ over $b$, Split Cycle says that $a$ defeats $b$. Similarly, since $b$ wins by more than $3$ over $c$, Split Cycle says that $b$ defeats $c$. However, since $c$ does not win by more than $3$ over $a$, Split Cycle says that $c$ does not defeat $a$. Crucially, though, since $c$ and $d$ are not involved in any cycles together, and $c$ wins by more than 0 over $d$, Split Cycle says that $c$ defeats $d$. The key point is that \textit{incoherence can be localized}: $c$ and $a$ belong to a cycle together, but $c$ and $d$ do not. By the same reasoning, $a$ defeats $d$, and $b$ defeats $d$. Thus, we obtain the following defeat relation:
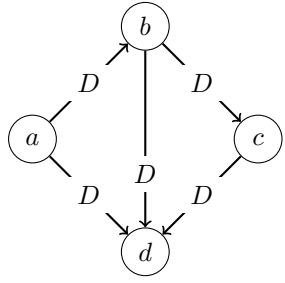
\begin{center}
\begin{minipage}{2in}\begin{tikzpicture}

\node[circle,draw, minimum width=0.25in] at (0,0) (a) {$a$}; 
\node[circle,draw,minimum width=0.25in] at (3,0) (c) {$c$}; 
\node[circle,draw,minimum width=0.25in] at (1.5,1.5) (b) {$b$}; 

\node[circle,draw,minimum width=0.25in] at (1.5,-1.5) (d) {$d$}; 

\path[->,draw,thick] (a) to node[fill=white] {$D$} (b);
\path[->,draw,thick] (b) to node[fill=white] {$D$} (c);

\path[->,draw,thick] (c) to node[fill=white] {$D$} (d);
\path[->,draw,thick] (b) to[pos=.7] node[fill=white] {$D$} (d);
\path[->,draw,thick] (a) to node[fill=white] {$D$} (d);

\end{tikzpicture}
\end{minipage}\end{center}
Note that just as in a sporting tournament, it can happen that while team $a$ defeats team $b$ and team $b$ defeats team $c$, team $a$ does not defeat team $c$, the same phenomenon occurs in the defeat relation above. Finally, since $a$ is the only undefeated candidate, $a$ is the winner according to Split Cycle.
\end{example}

\begin{remark}Where $f$ is the Split Cycle VCCR as in Definition \ref{WinByDef}, the induced voting method $\overline{f}$, which picks as winners the undefeated candidates, is the Split Cycle voting method. As a voting method, Split Cycle is Condorcet consistent: if $x$ is majority preferred to every other candidate $y$---if $x$ is a Condorcet winner---then $x$ is the unique winner of the election. For if $x$ is the Condorcet winner, then there are no cycles involving $x$, so $x$ defeats all other candidates according to Split Cycle.\end{remark}

An equivalent definition of Split Cycle can be given in terms of the following concept.

\begin{definition} Let $\mathbf{P}$ be a profile and $\rho$ a majority cycle in $\mathbf{P}$. The \textit{splitting number of $\rho$ in $\mathbf{P}$} is the smallest margin between consecutive candidates in $\rho$. Let $Split\#_\mathbf{P}(\rho)$ be the splitting number of $\rho$ in $\mathbf{P}$.
\end{definition}

For example, the splitting number of the cycle $a,b,c,a$ in the profile in Example \ref{WinByEx} is 3. In \citealt{HP2020}, we took the following formulation of Split Cycle to be the official definition (for a proof of Lemma \ref{SplittingLem}, see Appendix \ref{Proofs}).

\begin{restatable}{lemma}{Reformulation}\label{SplittingLem} Let $\mathbf{P}$ be a profile and $x,y\in X(\mathbf{P})$. Then $x$ defeats $y$ in $\mathbf{P}$ according to Split Cycle if and only if $Margin_\mathbf{P}(x,y)>0$ and
\[Margin_\mathbf{P}(x,y)>Split\#_\mathbf{P}(\rho)\mbox{ for every majority cycle $\rho$ in }\mathbf{P}\mbox{ containing $x$ and $y$}.\]
\end{restatable}

Thus, in Example \ref{WinByEx}, $a$ defeats $b$ because $Margin_\mathbf{P}(a,b)=5$, the only majority cycle is $a,b,c,a$, and its splitting number is $3$. Observe that since we are only comparing the sizes of margins, Split Cycle is qualitative-margin based in the sense of Remark \ref{QualMarg}.

\begin{example} It is important to note that two candidates may be contained in multiple majority cycles, as in the following margin graph, repeated three times with the majority cycles highlighted:
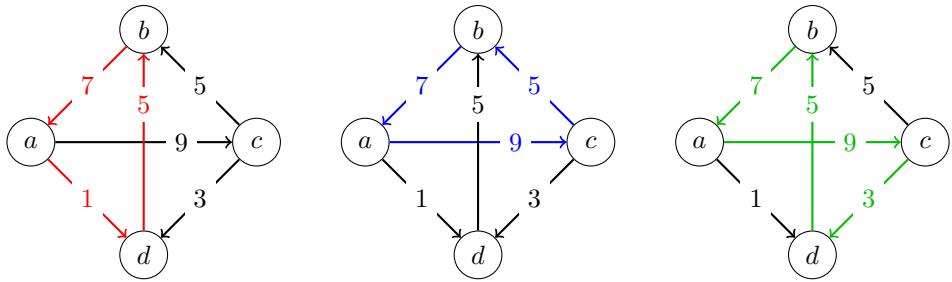
\begin{center}
\begin{minipage}{1.5in}\begin{tikzpicture}

\node[circle,draw, minimum width=0.25in] at (0,0) (a) {$a$}; 
\node[circle,draw,minimum width=0.25in] at (3,0) (c) {$c$}; 
\node[circle,draw,minimum width=0.25in] at (1.5,1.5) (b) {$b$}; 
\node[circle,draw,minimum width=0.25in] at (1.5,-1.5) (d) {$d$}; 

\path[->,draw,thick] (a) to (c);
\path[->,draw,thick,red] (d) to (b);
\path[->,draw,thick,red] (b) to node[fill=white] {$7$} (a);
\path[->,draw,thick] (c) to node[fill=white] {$5$} (b);
\path[->,draw,thick] (c) to node[fill=white] {$3$} (d);
\path[->,draw,thick,red] (a) to node[fill=white] {$1$} (d);

\node[fill=white] at (1.5,.5)  {{\color{red}$5$}}; 
\node[fill=white] at (2,0)  {$9$}; 

  \end{tikzpicture}
\end{minipage}\hspace{.25in}\begin{minipage}{1.5in}\begin{tikzpicture}

\node[circle,draw, minimum width=0.25in] at (0,0) (a) {$a$}; 
\node[circle,draw,minimum width=0.25in] at (3,0) (c) {$c$}; 
\node[circle,draw,minimum width=0.25in] at (1.5,1.5) (b) {$b$}; 
\node[circle,draw,minimum width=0.25in] at (1.5,-1.5) (d) {$d$}; 

\path[->,draw,thick,blue] (a) to (c);
\path[->,draw,thick] (d) to (b);
\path[->,draw,thick,blue] (b) to node[fill=white] {$7$} (a);
\path[->,draw,thick,blue] (c) to node[fill=white] {$5$} (b);
\path[->,draw,thick] (c) to node[fill=white] {$3$} (d);
\path[->,draw,thick] (a) to node[fill=white] {$1$} (d);

\node[fill=white] at (1.5,.5)  {$5$}; 
\node[fill=white] at (2,0)  {{\color{blue}$9$}}; 

  \end{tikzpicture}
\end{minipage}\hspace{.25in}\begin{minipage}{1.5in}\begin{tikzpicture}

\node[circle,draw, minimum width=0.25in] at (0,0) (a) {$a$}; 
\node[circle,draw,minimum width=0.25in] at (3,0) (c) {$c$}; 
\node[circle,draw,minimum width=0.25in] at (1.5,1.5) (b) {$b$}; 
\node[circle,draw,minimum width=0.25in] at (1.5,-1.5) (d) {$d$}; 

\path[->,draw,thick,medgreen] (a) to (c);
\path[->,draw,thick,medgreen] (d) to (b);
\path[->,draw,thick,medgreen] (b) to node[fill=white] {$7$} (a);
\path[->,draw,thick] (c) to node[fill=white] {$5$} (b);
\path[->,draw,thick,medgreen] (c) to node[fill=white] {$3$} (d);
\path[->,draw,thick] (a) to node[fill=white] {$1$} (d);

\node[fill=white] at (1.5,.5)  {{\color{medgreen}$5$}}; 
\node[fill=white] at (2,0)  {{\color{medgreen}$9$}}; 

  \end{tikzpicture}
\end{minipage}\end{center}
The splitting number of the cycle $a\to d \to b\to a$ is 1; the splitting number of the cycle $a\to c \to b\to a$ is 5; and the splitting number of the cycle $a\to c\to d\to b\to a$ is 3. Comparing the margins against these splitting numbers, one can calculate that the defeat relation is as follows:
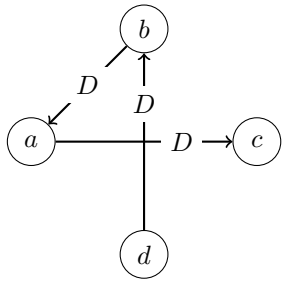
\begin{center}
\begin{minipage}{1.5in}\begin{tikzpicture}

\node[circle,draw, minimum width=0.25in] at (0,0) (a) {$a$}; 
\node[circle,draw,minimum width=0.25in] at (3,0) (c) {$c$}; 
\node[circle,draw,minimum width=0.25in] at (1.5,1.5) (b) {$b$}; 
\node[circle,draw,minimum width=0.25in] at (1.5,-1.5) (d) {$d$}; 

\path[->,draw,thick] (a) to (c);
\path[->,draw,thick] (d) to (b);
\path[->,draw,thick] (b) to node[fill=white] {$D$} (a);

\node[fill=white] at (1.5,.5)  {$D$}; 
\node[fill=white] at (2,0)  {$D$}; 

  \end{tikzpicture}
  \end{minipage}
  \end{center}
Since $d$ is the only undefeated candidate, $d$ is the winner according to Split Cycle.\footnote{To contrast this result with that of another VCCR, note that $c$ covers $d$ (left and right covering are equivalent in this case), whereas $a$, $b$, and $c$ are uncovered, so according to the Covering VCCRs, $c$ defeats $d$, whereas none of $a$, $b$, or $c$ is defeated.}

Additional examples of determining the Split Cycle defeat relation will be given below (Example \ref{IIAExample}, Remark \ref{MaskinRemark}, and Example \ref{BordaExample}). For still more examples, see \citealt{HP2020}.
\end{example}

A useful fact, proved in \citealt{HP2020}, is that it suffices to only look at majority cycles in which $y$ directly follows $x$. We include the proof in Appendix \ref{Proofs} to keep the paper self-contained.

\begin{restatable}{lemma}{OnlySome}\label{OnlySomeCycles} Let $\mathbf{P}$ be a profile and $x,y\in X(\mathbf{P})$. Then $x$ defeats $y$ in $\mathbf{P}$ according to Split Cycle if and only if $Margin_\mathbf{P}(x,y)>0$ and
\[Margin_\mathbf{P}(x,y)>Split\#_\mathbf{P}(\rho)\mbox{ for every majority cycle $\rho$ in }\mathbf{P}\mbox{ of the form } x \rightarrow y\rightarrow z_1\rightarrow \dots\rightarrow z_n\rightarrow x.\]
\end{restatable}

In \citealt{HP2020}, we show that Split Cycle---understood as a voting method, i.e., as $\overline{f}$ for the Split Cycle VCCR $f$---satisfies a number of desirable axioms for voting methods, and we systematically compare Split Cycle to other margin-based voting methods, including Beat Path and Ranked Pairs (\citealt{Tideman1987}). In particular, we show that Split Cycle is the only \textit{known} voting method that satisfies several axioms, including anti-spoiler axioms and axioms preventing the so-called Strong No Show Paradox. In the next two sections, we take a different approach: we characterize Split Cycle as a VCCR, rather than a voting method, and we characterize the Split Cycle VCCR relative to \textit{all} VCCRs.

\section{Axioms}\label{AxiomsSection}

In this section, we propose six axioms concerning when one candidate should defeat another in a democratic election involving two or more candidates. Four axiom are standard (Section \ref{StandardAx}); one is less well known but also from the previous literature (Section \ref{BalancedSection}); and the key axiom is new (Section \ref{CoherentIIASection}).

\subsection{Standard axioms}\label{StandardAx}

The first four axioms are ubiquitous in social choice and voting theory. The first axiom appears already in May's \citeyearpar{May1952} characterization of majority rule for two-candidate elections:
\begin{enumerate}
\item[A1.] Anonymity and Neutrality: if $x$ defeats $y$ in $\mathbf{P}$, and $\mathbf{P}'$ is obtained from $\mathbf{P}$ by swapping the ballots assigned to two voters, then $x$ still defeats $y$ in $\mathbf{P}'$ (Anonymity); and if $x$ defeats $y$ in $\mathbf{P}$, and $\mathbf{P}'$ is obtained from $\mathbf{P}$ by swapping $x$ and $y$ on each voter's ballot, then $y$ defeats $x$ in $\mathbf{P}'$ (Neutrality).\footnote{\label{Permutation}These versions of Anonymity and Neutrality stated in terms of the transposition of two ballots/candidates are equivalent to the usual versions stated in terms of a permutation of the ballots/candidates, since any permutation can be obtained by a sequence of transpositions. The usual version of, e.g., Neutrality states that if $\sigma$ is a permutation of $X$, and $\sigma \mathbf{P}$ is the profile obtained from $\mathbf{P}$ by setting $x \sigma\mathbf{P}_iy$ if and only if $\sigma(x)\mathbf{P}_i\sigma(y)$, then $x$ defeats $y$ in $\sigma\mathbf{P}$ if and only if $\sigma(x)$ defeats $\sigma(y)$ in $\mathbf{P}$.}
\end{enumerate}
It is clear that all VCCRs defined so far in this paper satisfy Anonymity and Neutrality.

The second axiom is definitive of the problem of choosing winners that we aim to solve:
\begin{enumerate}
\item[A2.] Availability: for every $\mathbf{P}$, there is some undefeated candidate in $\mathbf{P}$.
\end{enumerate}
To say that in some profiles all candidates are defeated and hence excluded from further consideration---so no candidate is available to become the ultimate winner---is to give up on solving the problem. Unlike the Simple Majority VCCR (Example \ref{OtherVCCRs}), Split Cycle satisfies Availability.

\begin{proposition}\label{NontrivialDefeat} Split Cycle satisfies Availability.
\end{proposition}
\begin{proof} Suppose there is a profile $\mathbf{P}$ in which every candidate is defeated by some other. Since $X(\mathbf{P})$ is finite, it follows that there is a sequence $\rho=x_1,\dots,x_n$ of candidates with $x_1=x_n$ such that each candidate defeats the next candidate in the sequence. It follows by Lemma \ref{SplittingLem} that $\rho$ is a majority cycle in which  the margin of each candidate over the next is greater than the splitting number of $\rho$, which is impossible. \end{proof}

\noindent Note that Availability is strictly weaker than the assumption that a VCCR is acyclic. There being no defeat cycles implies that some candidate is undefeated (given that the set of candidates in a profile is finite), but some candidate being undefeated does not imply that there are no defeat cycles (e.g., $a$ could be undefeated while there is a defeat cycle involving $b$, $c$, and $d$). Nonetheless, the proof of Proposition \ref{NontrivialDefeat} (starting in the second sentence) shows that Split Cycle is an acyclic VCCR as well.

For the third axiom, given any profile $\mathbf{P}$ and natural number $m$, the profile $m\mathbf{P}$ is obtained from $\mathbf{P}$ by replacing each voter by $m$ copies of that voter. For example, if $\mathbf{P}$ has three voters $i,j,k$, then $2\mathbf{P}$ has six voters $i_1,i_2,j_1,j_2,k_1,k_2$ such that the ballots of $i_1$ and $i_2$ in $2\mathbf{P}$ are the same as those of $i$ in $\mathbf{P}$, etc.
\begin{enumerate}
\item[A3.] (Upward) Homogeneity: for every $\mathbf{P}$, if $x$ defeat $y$ in $\mathbf{P}$, then $x$ defeats $y$ in $2\mathbf{P}$.
\end{enumerate}
Homogeneity is usually stated as the condition that for  any $m\geq 1$, $x$ defeats $y$ in $\mathbf{P}$ \textit{if and only if} $x$ defeats $y$ in $m\mathbf{P}$. As Smith \citeyearpar[p.~1029]{Smith1973} remarks, ``Homogeneity seems an extremely natural requirement; if each voter suddenly splits into $m$ voters, each of whom has the same preferences as the original, it would be hard to imagine how the ``collective preference'' would change.'' Nonetheless, we use the weaker version stated above since it is sufficient for our main result. Almost all standard voting procedures satisfy (the usual statement of) Homogeneity.\footnote{One exception is the Dodgson voting procedure (see \citealt{Fishburn1977}, \citealt{Brandt2009}).} That Split Cycle satisfies Homogeneity follows from the fact that Split Cycle is qualitative-margin based as in Remark \ref{QualMarg}, and $\mathbf{P}$ and $2\mathbf{P}$ have the same qualitative margin graphs.

The fourth axiom is one of the most widely discussed principles in voting theory. The term `Monotonicity' is used for a number of different conditions, but our formulation is equivalent (for profiles of linear ballots) to Arrow's \citeyearpar{Arrow1963} axiom of Positive Association of Social and Individual Values:\footnote{Positive Association states that if $x$ defeats $y$ in $\mathbf{P}$ according to $f$, and $\mathbf{P}'$ is a profile such that for all $x',y'\in X\setminus\{x\}$, (i)~$\mathbf{P}_{\mid \{x',y'\}}=\mathbf{P}'_{\mid \{x',y'\}}$, (ii) for all $i\in V$, $x \mathbf{P}_i y'$ implies $x \mathbf{P}'_i y'$, and (iii) for all $i\in V$, $y'\mathbf{P}'_ix$ implies $y'\mathbf{P}_ix$, then $x$ defeats $y$ in $\mathbf{P}'$ according to $f$.}
\begin{enumerate}
\item[A4.] Monotonicity (resp.~Monotonicity for two-candidate profiles): if $x$ defeats $y$ in a profile (resp.~two-candidate profile) $\mathbf{P}$, and $\mathbf{P}'$ is obtained from $\mathbf{P}$ by some voter $i$ moving $x$ above the candidate that $i$ ranked immediately above $x$ in $\mathbf{P}$, then $x$ defeats $y$ in $\mathbf{P}'$.
\end{enumerate}
One might argue that Monotonicity should hold for any number of candidates, but Monotonicity for two-candidate profiles is sufficient for the proof of our main result. This is noteworthy because the Instant Runoff VCCR in Example \ref{OtherVCCRs} (as well as VCCRs based on other standard voting procedures, e.g., Baldwin, Coombs, and Nanson) does not satisfy Monotonicity for arbitrary profiles (see  \citealt{Felsenthal2017}) but does for two-candidate profiles. All other VCCRs defined above satisfy Monotonicity for all profiles.

The axioms proposed so far imply the principle of Majority Defeat for two-candidates profiles. The proof is essentially part of the proof of May's \citeyearpar{May1952} characterization of majority rule.

\begin{lemma}\label{SpecialMaj} If $f$ satisfies Anonymity, Neutrality, and Monotonicity with respect to two-candidate profiles, then $f$ satisfies Special Majority Defeat: for any two-candidate profile $\mathbf{P}$,  $x$ defeats $y$  in $\mathbf{P}$ according to $f$ only if $x$ is majority preferred to $y$. 
\end{lemma}

\begin{proof} Suppose $x$ defeats $y$ in $\mathbf{P}$. It follows by Anonymity, Neutrality, and the asymmetry of defeat that the number of voters who rank $x$ above $y$ is not \textit{equal} to the number who rank $y$ above $x$. Now we claim that $x$ is majority preferred to $y$. Suppose instead that $y$ is majority preferred to $x$ by a margin of $m$. Flip voters with $y\mathbf{P}_ix$ to $x\mathbf{P}_i'y$ until we obtain a profile $\mathbf{P}'$ in which $x$ is majority preferred to $y$ by a margin of $m$. Since $x$ defeats $y$ in $\mathbf{P}$, $x$ still defeats $y$ in $\mathbf{P}'$ by Monotonicity. But $\mathbf{P}'$ can also be obtained from $\mathbf{P}$ by the voter and candidate swaps described in the statements of Anonymity and Neutrality. Thus, since $x$ defeats $y$ in $\mathbf{P}$,  $y$ defeats $x$ in $\mathbf{P}'$. Hence in $\mathbf{P}'$, $x$ defeats $y$ and $y$ defeats $x$, contradicting the asymmetry of defeat.\end{proof}

\begin{remark}\label{PRremark} Monotonicity is weaker than May's \citeyearpar{May1952} condition of Positive Responsiveness, which in addition requires that if $y$ does not defeat $x$ in $\mathbf{P}$, then changing a single voter from $y\mathbf{P}_ix$ to $x\mathbf{P}'_iy$ results in $x$ defeating $y$ in $\mathbf{P}'$. We find imposing Positive Responsiveness in general, i.e., for elections with any number of candidates, much too strong,\footnote{Cf.~Woeginger \citeyearpar{Woeginger2003}, who notes that while Anonymity and Neutrality ``are natural and fairly weak, the positive responsiveness axiom is usually criticized for being too strong'' (p.~89).} so we prefer to motivate majority rule in two-candidate elections using axioms that are more plausible for any number of candidates, as in Proposition \ref{Majority} below.\end{remark}

\subsection{Neutral Reversal}\label{BalancedSection}

Like the axiom of Homogeneity,  the next axiom is a variable-voter axiom. Say that two voters $i$ and $j$ have \textit{reversed ballots} in a profile $\mathbf{Q}$ if for all $x,y\in X(\mathbf{Q})$, we have $x\mathbf{Q}_iy$ if and only if $y\mathbf{Q}_j x$. For example, if $i$ has  $abcd$ and $j$ has $dcba$, then $i$ and $j$ have reversed ballots. Adding a pair of voters with reversed ballots to a profile does not change the margins between any candidates. A natural thought is that such voters balance each other out, so adding such a pair to an election should not change the defeat relations between candidates. This leads to what Saari \citeyearpar{Saari2003} calls the Neutral Reversal Requirement.

\begin{itemize}
\item[A5.] Neutral Reversal: if $\mathbf{P}'$ is obtained  from $\mathbf{P}$ by adding two voters with reversed ballots, then $x$ defeats $y$ in $\mathbf{P}$ if and only if $x$ defeats $y$ in $\mathbf{P}'$.
\end{itemize}
Not only Split Cycle but all other margin-based VCCRs (recall Example \ref{OtherVCCRs}) satisfy Neutral Reversal. However, Neutral Reversal is weaker than the assumption that a VCCR is margin based.

\begin{example} We define the Positive/Negative VCCR (cf.~\citealt{Lapresta2010} and \citealt{Heckelman2020}) as follows. In a profile $\mathbf{P}$, a candidate $x$ receives $1$ point for every voter who ranks $x$ first and $-1$ point for every voter who ranks $x$ last. The score of $x$ in $\mathbf{P}$ is the sum of the points $x$ receives from voters. Then $x$ defeats $y$ in $\mathbf{P}$ if the score of $x$ is greater than the score of $y$. This Positive/Negative VCCR satisfies Neutral Reversal, as a pair of reversed ballots adds a net score of 0 to each candidate. Yet it is easy to construct profiles with the same margin graphs that have different defeat relations.\end{example}

Not all common VCCRs satisfy Neutral Reversal, as Examples \ref{PluralityReversal}-\ref{ParetoReversal} below show. To analyze violations of Neutral Reversal, we distinguish its two directions:
\begin{itemize}
\item Upward Neutral Reversal: if $\mathbf{P}'$ is obtained  from $\mathbf{P}$ by adding two voters with reversed ballots, then if $x$ defeats $y$ in $\mathbf{P}$, $x$ defeats $y$ in $\mathbf{P}'$.
\item Downward Neutral Reversal: if $\mathbf{P}'$ is obtained  from $\mathbf{P}$ by adding two voters with reversed ballots, then if $x$ defeats $y$ in $\mathbf{P}'$, $x$ defeats $y$ in $\mathbf{P}$.
\end{itemize}

\begin{example}\label{PluralityReversal} Consider the Plurality VCCR from Example \ref{OtherVCCRs}. Let $\mathbf{P}$ be any profile for candidates $a,b,c$. Adding to $\mathbf{P}$ a pair of voters with the reversed ballots $abc$ and $cba$ to obtain a profile $\mathbf{P}'$ increases the plurality scores of $a$ and $c$ by one but does not increase the plurality score of $b$. From here it is easy to see that the Plurality VCCR violates both Upward and Downward Neutral Reversal.\end{example}

\begin{example}\label{ParetoReversal} The Pareto VCCR $f$ is defined as follows: for any profile $\mathbf{P}$ and $x,y\in X(\mathbf{P})$, $x$ defeats $y$ in $\mathbf{P}$ if and only if all voters in $\mathbf{P}$ rank $x$ above $y$. Clearly adding two voters with reversed ballots to a profile in which $x$ is unanimously ranked above $y$ results in a profile in which $x$ is not unanimously ranked above $y$. Thus, the Pareto VCCR violates Upward Neutral Reversal. However, it trivially satisfies Downward Neutral Reversal, because if $\mathbf{P}'$ has a pair of voters with reversed ballots, then no candidates defeats any other in $\mathbf{P}'$.\end{example}

The Pareto VCCR seems reasonable for certain special purposes, e.g., in a small club, unanimity may be valued and often possible. However, in elections where disagreement is expected, the Pareto VCCR would be of little help in narrowing down the range of potential winners, as so few candidates would defeat others. Of course, we agree that it is a \textit{sufficient} condition for $x$ to defeat $y$ that $x$ is unanimously ranked above $y$. A VCCR $f$ is said to satisfy the Pareto axiom if for all profiles $\mathbf{P}$ and $x,y\in X(\mathbf{P})$, if $x\mathbf{P}_iy$ for all $i\in V(\mathbf{P})$, then $x$ defeats $y$ in $\mathbf{P}$ according to $f$. Split Cycle clearly satisfies the Pareto axiom, since there cannot be a majority cycle all of whose margins are equal to the total number of voters. Moreover, we can use Pareto and Upward Neutral Reversal to derive the converse of  Special Majority Defeat (Lemma \ref{SpecialMaj}), thereby obtaining a characterization of majority rule for two-candidate elections that differs from May's \citeyearpar{May1952} famous characterization (cf.~\citealt{Asan2002}, \citealt{Woeginger2003}, and \citealt[Cor.~12]{Llamazares2006}).

\begin{proposition}\label{Majority} For any VCCR $f$, the following are equivalent:
\begin{enumerate}
\item\label{Majority1} $f$ coincides with majority rule on two-candidate profiles;
\item\label{Majority2} $f$ satisfies the following axioms with respect to two-candidate profiles: Anonymity, Neutrality, Monotonicity, Pareto, and Upward Neutral Reversal.
\end{enumerate}
\end{proposition}

\begin{proof} The implication from 1 to 2 is easy to check. From 2 to 1, we already proved in Lemma \ref{SpecialMaj} that if $f$ satisfies Anonymity, Neutrality, and Monotonicity for two-candidate profiles, then in such a profile, $x$ defeats $y$ only if $x$ is majority preferred to $y$. We now use Pareto and Upward Neutral Reversal to show that if $x$ is majority preferred to $y$, then $x$ defeats $y$. For suppose $\mathbf{P}$ is a two-candidate profile in which $x$ is majority preferred to $y$. Consecutively remove pairs of voters $(i,j)$ with $x\mathbf{P}_iy$ and $y\mathbf{P}_jx$ until we obtain a profile $\mathbf{P}_0$ in which all voters rank $x$ over $y$.  By Pareto, $x$ defeats $y$ in $\mathbf{P}_0$. Then by repeated application of Upward Neutral Reversal, adding back the removed pairs of voters, $x$ defeats $y$ in $\mathbf{P}$.
\end{proof}

\noindent We prefer this characterization of majority rule to that of May \citeyearpar{May1952} for the reason given in Remark \ref{PRremark}.

\subsection{Coherent IIA and The Fallacy of IIA}\label{CoherentIIASection}

Suppose $x$ defeats $y$ in a profile $\mathbf{P}$, and a profile $\mathbf{P}'$ is exactly like $\mathbf{P}$ with respect to how every voter ranks $x$ vs.~$y$. Should it follow that $x$ defeats $y$ in $\mathbf{P}'$? Arrow's \citeyearpar{Arrow1963} famous axiom of the Independence of Irrelevant Alternatives (IIA) says `yes' (see Section \ref{FixedImposs}). But we say `no' if $\mathbf{P}'$ is more incoherent than $\mathbf{P}$, in terms of cycles. If $\mathbf{P}'$ is sufficiently incoherent, we may need to suspend judgment on many defeat relations that we could coherently accept in $\mathbf{P}$. To overlook this point is to commit what we call The Fallacy of IIA.\footnote{We take this criticism of IIA to differ from some other criticisms of IIA, such as those in \citealt[Ch.~6]{Mackie2003}.}

Although there is a perfectly reasonable notion of the \textit{advantage} of  $x$ over  $y$ that only depends on how voters rank $x$ vs.~$y$, whether that intrinsic advantage is sufficient for $x$ to \textit{defeat} $y$ may depend on a standard that takes into account the whole election, e.g., that takes into account whether the electorate is incoherent with respect to a set of candidates including $x,y$ (see \citealt{HK2020b} for a formalization of the advantage-standard idea). That standards may be \textit{context dependent} should be no surprise: just as whether a person counts as ``tall'' depends on who else is being assessed for tallness in the context of our judgment, whether one candidate's performance against another counts as ``a defeat'' depends on which other pairwise candidate performances are also being assessed as potential defeats in the context of our judgment.\footnote{For a related proposal in the setting of judgment aggregation to set supermajority thresholds in a local, context-sensitive way, see \citealt{Cariani2016}.}

\begin{example}\label{IIAExample} In the profiles $\mathbf{P}$ and $\mathbf{P}'$ below, we have $\mathbf{P}_{\mid \{a,b\}}=\mathbf{P}'_{\mid \{a,b\}}$. In the context of the perfectly coherent profile $\mathbf{P}$, the margin of $n$ for $a$ over $b$ should be sufficient for $a$ to defeat $b$. But in the context of the incoherent profile $\mathbf{P}'$, it is not sufficient: no one can be judged to defeat anyone else (this follows from Anonymity, Neutrality, and Availability). Thus, this is a counterexample to IIA.

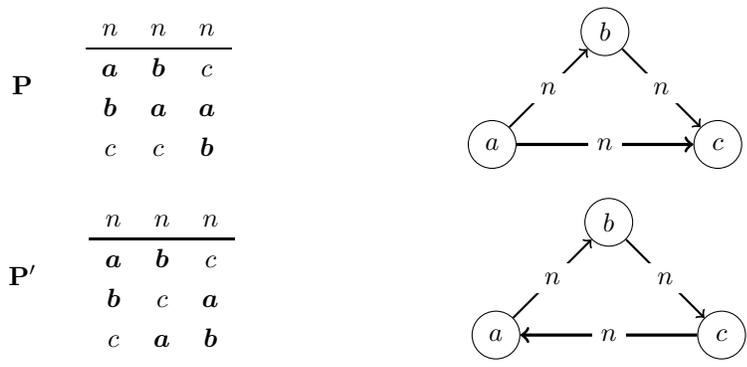
\begin{center}
$\mathbf{P}$\qquad
\begin{minipage}{2in}\begin{tabular}{ccc}
$n$ & $n$ & $n$   \\\hline
$\boldsymbol{a}$ & $\boldsymbol{b}$ &  $c$ \\
$\boldsymbol{b}$ &  $\boldsymbol{a}$ & $\boldsymbol{a}$ \\
$c$ &  $c$ &  $\boldsymbol{b}$ \\
\end{tabular}\end{minipage}\begin{minipage}{2in}\begin{tikzpicture}

\node[circle,draw, minimum width=0.25in] at (0,0) (a) {$a$}; 
\node[circle,draw,minimum width=0.25in] at (3,0) (c) {$c$}; 
\node[circle,draw,minimum width=0.25in] at (1.5,1.5) (b) {$b$}; 

\path[->,draw,thick] (a) to node[fill=white] {$n$} (b);
\path[->,draw,thick] (b) to node[fill=white] {$n$} (c);
\path[->,draw,very thick] (a) to node[fill=white] {$n$} (c);

\end{tikzpicture}
\end{minipage}\end{center}

\begin{center}
$\mathbf{P}'$\qquad
\begin{minipage}{2in}\begin{tabular}{ccc}
$n$ & $n$ & $n$   \\\hline
$\boldsymbol{a}$ & $\boldsymbol{b}$ &  $c$ \\
$\boldsymbol{b}$ &  $c$ & $\boldsymbol{a}$ \\
$c$ &  $\boldsymbol{a}$ &  $\boldsymbol{b}$ \\
\end{tabular}\end{minipage}\begin{minipage}{2in}\begin{tikzpicture}

\node[circle,draw, minimum width=0.25in] at (0,0) (a) {$a$}; 
\node[circle,draw,minimum width=0.25in] at (3,0) (c) {$c$}; 
\node[circle,draw,minimum width=0.25in] at (1.5,1.5) (b) {$b$}; 

\path[->,draw,thick] (a) to node[fill=white] {$n$} (b);
\path[->,draw,thick] (b) to node[fill=white] {$n$} (c);
\path[->,draw,very thick] (c) to node[fill=white] {$n$} (a);

\end{tikzpicture}
\end{minipage}\end{center}
However, it is not as if the standard for defeat in every case of a majority cycle is unattainable. In the profile $\mathbf{Q}$ below, we believe that the advantage of $a$ over $b$ is sufficient for $a$ to defeat $b$:

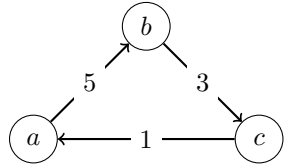
\begin{center}
$\mathbf{Q}$\qquad
\begin{minipage}{2in}\begin{tabular}{ccc}
$4$ & $2$ & $3$   \\\hline
$a$ & $b$ &  $c$ \\
$b$ &  $c$ & $a$ \\
$c$ &  $a$ &  $b$ \\
\end{tabular}\end{minipage}\begin{minipage}{2in}\begin{tikzpicture}

\node[circle,draw, minimum width=0.25in] at (0,0) (a) {$a$}; 
\node[circle,draw,minimum width=0.25in] at (3,0) (c) {$c$}; 
\node[circle,draw,minimum width=0.25in] at (1.5,1.5) (b) {$b$}; 

\path[->,draw,thick] (b) to node[fill=white] {$3$} (c);
\path[->,draw,thick] (c) to node[fill=white] {$1$} (a);
\path[->,draw,thick] (a) to node[fill=white] {$5$} (b);

\end{tikzpicture}
\end{minipage}\end{center}
According to Split Cycle, the standard for $a$ to defeat $b$ in a profile, which the margin of $a$ over $b$ must surpass, is the maximum of the splitting numbers of the cycles containing $a$ and $b$. Since the splitting number of the cycle in the profile $\mathbf{Q}$ is 1, the margin of $a$ over $b$ surpasses the standard, so $a$ defeats $b$.\end{example}

\begin{remark} In an illuminating result, Patty and Penn \citeyearpar{Patty2014} prove that Arrow's IIA is equivalent to the condition of \textit{unilateral flip independence}, which states that if two profiles are alike except that one voter flips one pair of adjacent candidates on her ballot, then the defeat relations for the two profiles can differ at most on the flipped candidates. They write that this theorem ``demonstrates a fundamental basis of the normative appeal of IIA'' (p.~52) (cf.~\citealt[p.~155]{Patty2019}).\footnote{Patty and Penn \citeyearpar{Patty2019} do not think that IIA is compelling for voting procedures in elections (see their Section 3.1), but they do find it compelling in contexts of multicriterial decision making.} However, observe that for the profiles $\mathbf{P}$ and $\mathbf{P}'$ in Example \ref{IIAExample}, if $n=1$, then a single voter (the middle voter) flipping adjacent candidates on her ballot ($ac$ to $ca$) takes us from the coherent profile $\mathbf{P}$, in which there is no difficulty in judging that $a$ defeats $b$, to the incoherent profile $\mathbf{P}'$, in which no one can be judged to defeat anyone else. Hence unilateral flip independence makes the same mistake as IIA in ignoring how context can affect the standard for defeat.\end{remark}

\begin{remark}\label{MaskinRemark} Maskin \citeyearpar{Maskin2020} proposes a weakening of IIA called Modified IIA, which states that if profiles $\mathbf{P}$ and $\mathbf{P}'$ are alike in how every voter ranks $x$ vs.~$y$, and for each voter $i$ and candidate $z$, $i$ ranks $z$ in between $x$ and $y$ in $\mathbf{P}$ if and only if $i$ ranks $z$ in between $x$ and $y$ in $\mathbf{P}'$, then $x$ defeats $y$ in $\mathbf{P}$ if and only if $x$ defeats $y$ in $\mathbf{P}'$. Saari \citeyearpar{Saari1994,Saari1995,Saari1998} proposed a stronger axiom, though still weaker than IIA, called Intensity IIA: if profiles $\mathbf{P}$ and $\mathbf{P}'$ are alike in how every voter ranks $x$ vs.~$y$, and for each voter $i$, the number of candidates whom $i$ ranks in between $x$ and $y$ in $\mathbf{P}$ is the same as the number of candidates whom $i$ ranks in between $x$ and $y$ in $\mathbf{P}'$, then $x$ defeats $y$ in $\mathbf{P}$ if and only if $x$ defeats $y$ in $\mathbf{P}'$. Modified IIA and Intensity IIA are problematic for the same reason that IIA is, only we now need four candidates to see why. In the profiles $\mathbf{P}$ and $\mathbf{P}'$ below, we have $\mathbf{P}_{\mid \{a,b\}}=\mathbf{P}'_{\mid \{a,b\}}$, and for each voter $i$ and candidate $z$, $i$ ranks $z$ in between $a$ and $b$ in $\mathbf{P}$ if and only if $i$ ranks $z$ in between $a$ and $b$ in $\mathbf{P}'$. In the context of the perfectly coherent profile $\mathbf{P}$, the margin of $2n$ for $a$ over $b$ should be sufficient for $a$ to defeat $b$. But in the context of the incoherent profile $\mathbf{P}'$, it is not sufficient: no one can be judged to defeat anyone else (this follows from Anonymity, Neutrality, and Availability). Thus, this is a counterexample to Modified IIA and Intensity~IIA.
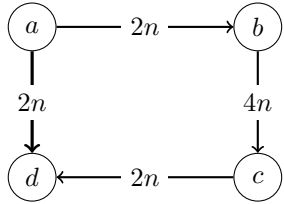
\begin{center}
$\mathbf{P}$\qquad
\begin{minipage}{2in}\begin{tabular}{cccc}
$n$ & $n$ & $n$ & $n$  \\\hline
$\boldsymbol{a}$ & $\boldsymbol{b}$ &  $\boldsymbol{a}$ & $\boldsymbol{a}$ \\
$\boldsymbol{b}$ &  $c$ & $\boldsymbol{b}$ & $\boldsymbol{b}$ \\
$c$ &  $d$ &  $c$ & $d$ \\
$d$ &  $\boldsymbol{a}$ & $d$ & $c$
\end{tabular}\end{minipage}\begin{minipage}{2in}\begin{tikzpicture}

\node[circle,draw, minimum width=0.25in] at (0,2) (a) {$a$}; 
\node[circle,draw,minimum width=0.25in] at (3,2) (b) {$b$};
\node[circle,draw,minimum width=0.25in] at (3,0) (c) {$c$}; 
 \node[circle,draw,minimum width=0.25in] at (0,0) (d) {$d$}; 
 
\path[->,draw,thick] (a) to node[fill=white] {$2n$} (b);
\path[->,draw,thick] (b) to node[fill=white] {$4n$} (c);
\path[->,draw,thick] (c) to node[fill=white] {$2n$} (d);
\path[->,draw,very thick] (a) to node[fill=white] {$2n$} (d);

\path[->,draw,thick] (a) to node[fill=white,pos=0.75] {$2n$} (c);
\path[->,draw,thick] (b) to node[fill=white,pos=0.75] {$4n$} (d);

\end{tikzpicture}
\end{minipage}\end{center}

\begin{center}
$\mathbf{P}'$\qquad
\begin{minipage}{2in}\begin{tabular}{cccc}
$n$ & $n$ & $n$ & $n$  \\\hline
$\boldsymbol{a}$ & $\boldsymbol{b}$ &  $c$ & $d$ \\
$\boldsymbol{b}$ &  $c$ & $d$ & $\boldsymbol{a}$ \\
$c$ &  $d$ &  $\boldsymbol{a}$ & $\boldsymbol{b}$ \\
$d$ &  $\boldsymbol{a}$ & $\boldsymbol{b}$ & $c$
\end{tabular}\end{minipage}\begin{minipage}{2in}
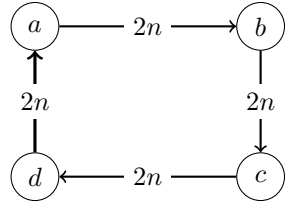
\vspace{.1in}\begin{tikzpicture}

\node[circle,draw, minimum width=0.25in] at (0,2) (a) {$a$}; 
\node[circle,draw,minimum width=0.25in] at (3,2) (b) {$b$};
\node[circle,draw,minimum width=0.25in] at (3,0) (c) {$c$}; 
 \node[circle,draw,minimum width=0.25in] at (0,0) (d) {$d$}; 
 
\path[->,draw,thick] (a) to node[fill=white] {$2n$} (b);
\path[->,draw,thick] (b) to node[fill=white] {$2n$} (c);
\path[->,draw,thick] (c) to node[fill=white] {$2n$} (d);
\path[->,draw,very thick] (d) to node[fill=white] {$2n$} (a);

\end{tikzpicture}
\end{minipage}\end{center}
Maskin \citeyearpar{Maskin2020} suggests the benefit of Modified IIA is that it rules out vote-splitting, which he illustrates using spoiler effects in Plurality voting as in Examples \ref{BushGore} and \ref{Trump}. However, Modified IIA is neither necessary nor sufficient for a voting procedure to have good anti-spoiler properties. Split Cycle does not satisfy Modified IIA---it correctly says that $a$ defeats $b$ in $\mathbf{P}$ but not $\mathbf{P}'$ above---yet Split Cycle satisfies strong anti-spoiler properties: not only Independence of Clones (\citealt{Tideman1987}), as shown in \citealt{HP2020}, but also a condition of Immunity to Spoilers, as shown below. On the other hand, Borda  satisfies Modified IIA yet satisfies neither Independence of Clones nor Immunity to Spoilers, as shown in Example \ref{BorderSpoiler} below.\end{remark}

Avoiding The Fallacy of IIA does not mean abandoning the idea behind IIA entirely.\footnote{For an intellectual history of the ideas behind IIA, going back to Condorcet and Daunou, see \citealt{McLean1995b}.} We need only depart from its local evaluation of $x$ vs.~$y$ when increasing incoherence demands that we be more conservative in locking in relations of defeat.  If there is no increase in incoherence from profile $\mathbf{P}$ to $\mathbf{P}'$, then if the intrinsic advantage of $x$ over $y$ is sufficient for $x$ to defeat $y$ in $\mathbf{P}$, we think it should still be sufficient for $x$ to defeat $y$ in $\mathbf{P}'$.  Moreover, a clearly sufficient condition for there to be no increase in incoherence from $\mathbf{P}$ to $\mathbf{P}'$ is the following: the margin graph of $\mathbf{P}'$ is obtained from that of $\mathbf{P}$ by deleting zero or more candidates other than $x$ and $y$ and deleting or reducing the margins on zero or more edges not connecting $x$ and $y$.\footnote{By allowing for the deletion of candidates, Coherent IIA is a weakening of a variable-candidate version of IIA that we call VIIA, defined in Section \ref{IIAVCCR}.} For such deletions or reductions can only \textit{reduce} incoherence. For example, in the profile $\mathbf{Q}$ in Example \ref{IIAExample}, deleting candidate $c$ or deleting or reducing the margins on the $c\to a$ or $b\to c$ edges only reduces incoherence, so $a$'s defeat of $b$ should be preserved. Thus, we arrive at our proposed axiom of Coherent IIA.

\begin{itemize}
\item[A6.] Coherent IIA: if $x$ defeats $y$ in $\mathbf{P}$, and $\mathbf{P}'$ is a profile such that $\mathbf{P}_{\mid \{x,y\}}=\mathbf{P}'_{\mid \{x,y\}}$ and the margin graph of $\mathbf{P}'$ is obtained from that of $\mathbf{P}$ by deleting zero or more candidates other than $x$ and $y$ and deleting or reducing the margins on zero or more edges not connecting $x$ and $y$, then $x$ still defeats $y$ in $\mathbf{P}'$.
\end{itemize}
In Section \ref{BaigentSection}, we show that Coherent IIA implies Weak IIA (\citealt{Baigent1987}): if two profiles are exactly alike with respect to how every voter ranks $x$ vs.~$y$, it cannot be that in one profile $x$ defeats $y$ while in the other profile $y$ defeats $x$. At most, a defeat that holds in one can be withdrawn in the other.

\begin{proposition}\label{CoherentProp} Split Cycle satisfies Coherent IIA.
\end{proposition}
\begin{proof}Suppose $x$ defeats $y$ in $\mathbf{P}$, so by Lemma \ref{SplittingLem}, $Margin_\mathbf{P}(x,y)>Split\#_\mathbf{P}(\rho)$ for every  majority cycle $\rho$ in $\mathbf{P}$ containing $x$ and $y$. Since $\mathbf{P}_{\mid \{x,y\}}=\mathbf{P}'_{\mid \{x,y\}}$,  $Margin_\mathbf{P}(x,y)=Margin_{\mathbf{P}'}(x,y)$. Since $\mathcal{M}(\mathbf{P}')$ is obtained from $\mathcal{M}(\mathbf{P})$ by deleting zero or more candidates other than $x$ and $y$ and zero or more edges not connecting $x$ and $y$, every majority cycle $\rho$ in $\mathbf{P}'$ containing $x$ and $y$ is already a majority cycle in $\mathbf{P}$ containing $x$ and $y$, and as no margins have increased from $\mathbf{P}$ to $\mathbf{P}'$,  $Split\#_{\mathbf{P}}(\rho)\geq Split\#_{\mathbf{P}'}(\rho)$. It follows that $Margin_{\mathbf{P}'}(x,y)>Split\#_{\mathbf{P}'}(\rho)$ for every majority cycle $\rho$ in $\mathbf{P}'$ containing $x$ and $y$. Thus, $x$ defeats $y$ in~$\mathbf{P}'$ by Lemma \ref{SplittingLem}.\end{proof}
 
Unlike our proposed axioms in previous sections, Coherent IIA significantly cuts down the space of allowable VCCRs, ruling out all the VCCRs in Examples \ref{OtherVCCRs} and \ref{BeatPathEx} except for Simple Majority. For example, it rules out the Borda VCCR as follows.

\begin{example}\label{BordaExample} To see that Borda fails Coherent IIA, consider the following profiles $\mathbf{P}$  and $\mathbf{P}'$:

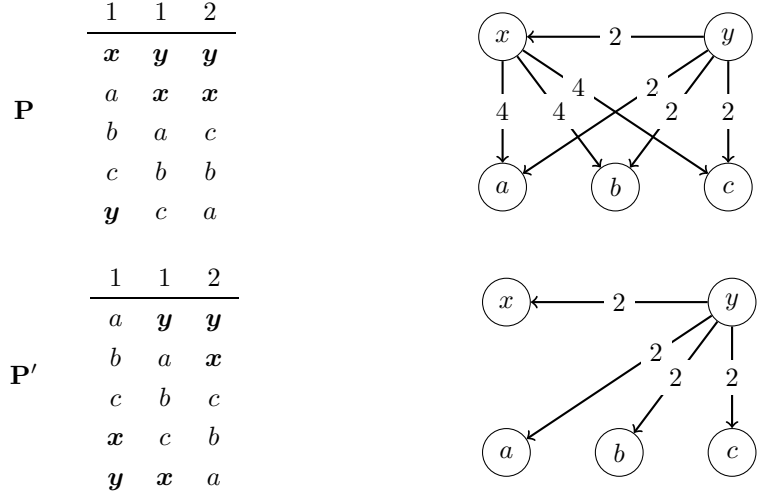
  \begin{center}
$\mathbf{P}$\qquad
 \begin{minipage}{2in}\begin{tabular}{ccc}
 $1$ & $1$ & $2$    \\\hline
$\boldsymbol{x}$ & $\boldsymbol{y}$ & $\boldsymbol{y}$    \\
$a$ &  $\boldsymbol{x}$ & $\boldsymbol{x}$   \\
$b$ &  $a$ & $c$  \\
$c$ &  $b$ & $b$  \\
$\boldsymbol{y}$ &  $c$ & $a$ \\
\end{tabular}
\end{minipage} \begin{minipage}{2in}\begin{tikzpicture}

\node[circle,draw, minimum width=0.25in] at (0,1) (x) {$x$}; 
\node[circle,draw,minimum width=0.25in] at (0,-1) (a) {$a$}; 
\node[circle,draw,minimum width=0.25in] at (3,1) (y) {$y$}; 
\node[circle,draw,minimum width=0.25in] at (3,-1) (c) {$c$}; 

\node[circle,draw,minimum width=0.25in] at (1.5,-1) (b) {$b$}; 

\path[->,draw,thick] (y) to node[fill=white] {$2$} (x);

\path[->,draw,thick] (x) to node[fill=white] {$4$}  (a);
\path[->,draw,thick] (x) to node[fill=white] {$4$} (b);
\path[->,draw,thick] (x) to[pos=0.3] node[fill=white] {$4$} (c);

\path[->,draw,thick] (y) to[pos=0.3] node[fill=white] {$2$} (a);
\path[->,draw,thick] (y) to node[fill=white] {$2$} (b);
\path[->,draw,thick] (y) to node[fill=white] {$2$} (c);

 \end{tikzpicture}
\end{minipage}

\end{center}

\begin{center}
$\mathbf{P}'$\qquad
 \begin{minipage}{2in}\begin{tabular}{ccc}
 $1$ & $1$ & $2$     \\\hline
$a$ & $\boldsymbol{y}$ & $\boldsymbol{y}$     \\
$b$ &  $a$ & $\boldsymbol{x}$   \\
$c$ &  $b$ & $c$   \\
$\boldsymbol{x}$ &  $c$ & $b$  \\
$\boldsymbol{y}$ &  $\boldsymbol{x}$ & $a$  \\
\end{tabular}\end{minipage}  \begin{minipage}{2in}\begin{tikzpicture}

\node[circle,draw, minimum width=0.25in] at (0,1) (x) {$x$}; 
\node[circle,draw,minimum width=0.25in] at (0,-1) (a) {$a$}; 
\node[circle,draw,minimum width=0.25in] at (3,1) (y) {$y$}; 
\node[circle,draw,minimum width=0.25in] at (3,-1) (c) {$c$}; 

\node[circle,draw,minimum width=0.25in] at (1.5,-1) (b) {$b$}; 

\path[->,draw,thick] (y) to node[fill=white] {$2$} (x);

\path[->,draw,thick] (y) to[pos=0.3] node[fill=white] {$2$} (a);
\path[->,draw,thick] (y) to node[fill=white] {$2$} (b);
\path[->,draw,thick] (y) to node[fill=white] {$2$} (c);

 \end{tikzpicture}
\end{minipage}
\end{center}
According to the Borda VCCR, $x$ defeats $y$ in $\mathbf{P}$: despite the fact that only one person prefers $x$ to $y$, whereas \textit{three} prefer $y$ to $x$, the proponent of Borda ascribes some significance to the fact that the first voter places $a,b,c$ between $x$ and $y$ (perhaps strategically, of course). Now although $\mathbf{P}_{\mid \{x,y\}}=\mathbf{P}'_{\mid \{x,y\}}$ and the margin graph of $\mathbf{P}'$ is obtained from that of $\mathbf{P}$ by deleting some edges not connecting $x$ and $y$, Borda changes the verdict for $\mathbf{P}'$ and says that $y$ defeats $x$ in $\mathbf{P}'$. Thus, Borda violates Coherent IIA. The mistake, in our view, was to judge that $x$ defeats $y$ in $\mathbf{P}$ in the first place. According to Split Cycle, by contrast, $y$ defeats $x$ in $\mathbf{P}$ because $y$ is majority preferred to $x$ and there are no majority cycles in $\mathbf{P}$.\end{example}

Example \ref{BordaExample} shows that Borda fails to satisfy the principle of Majority Defeat from Section \ref{IntroDefeat}: if $x$ defeats $y$ in $\mathbf{P}$, then $x$ is majority preferred to $y$ in $\mathbf{P}$. 

\begin{lemma}\label{MajDefeatLem} Anonymity, Neutrality,  Monotonicity (for two-candidate profiles), and Coherent IIA together imply Majority Defeat.
\end{lemma}
\begin{proof} Suppose $x$ defeats $y$ in $\mathbf{P}$. Then by Coherent IIA, $x$ defeats $y$ in $\mathbf{P}_{\mid \{x,y\}}$. Given Anonymity, Neutrality, and Monotonicity, it follows by Lemma \ref{SpecialMaj} that $x$ is majority preferred to $y$ in  $\mathbf{P}_{\mid \{x,y\}}$ and hence in~$\mathbf{P}$.\end{proof}

\begin{remark}\label{MajCon} Whenever an election has a Condorcet winner, Majority Defeat implies that the Condorcet winner is undefeated, but it does not imply that the Condorcet winner is the \textit{only} undefeated candidate. Thus, Majority Defeat does not by itself imply Condorcet consistency.\end{remark}

Finally, we will show that Coherent IIA (together with Anonymity, Neutrality, and Monotonicity) rules out the kind of spoiler effects shown in Examples \ref{BushGore}-\ref{Trump}. For this, we must consider an election with and without a potential spoiler. Given a profile $\mathbf{P}$ and $b\in X(\mathbf{P})$, let $\mathbf{P}_{-b}$ be the profile obtained from $\mathbf{P}$ by deleting $b$ from all ballots.\footnote{I.e., $\mathbf{P}_{-b}=\mathbf{P}_{\mid X(\mathbf{P})\setminus\{b\}}$, using the notation of Definition \ref{Restriction}.} In \citealt{HP2020}, we define the following  axiom:
\begin{itemize}
\item Immunity to Spoilers: if $a$ is undefeated in $\mathbf{P}_{-b}$, and $a$ is majority preferred to $b$ in $\mathbf{P}$, and $b$ is defeated in $\mathbf{P}$, then $a$ is still undefeated in $\mathbf{P}$.
\end{itemize}

\noindent Examples \ref{BushGore} and \ref{Trump} show how Plurality voting can violate Immunity to Spoilers. In the first example, assume Gore would win in the two-candidate profile $\mathbf{P}_{-\mathrm{Nader}}$ and that Gore is majority preferred to Nader in the full election $\mathbf{P}$. Then since Nader is defeated in $\mathbf{P}$, Immunity to Spoilers requires that Nader not spoil the election for Gore, i.e., that Gore is still a winner in $\mathbf{P}$. The analysis of Example \ref{Trump} is similar, making some assumptions about voters' rankings of Cruz, Kasich, and Rubio.  Example \ref{Burlington} shows how Instant Runoff voting can violate Immunity to Spoilers. Montroll wins in the two-candidate profile $\mathbf{P}_{-\mathrm{Wright}}$, and Montroll is majority preferred to Wright in the full election $\mathbf{P}$. Then since Wright is defeated  in $\mathbf{P}$, Immunity to Spoilers requires that Wright not spoil the election for Montroll, i.e., that Montroll is still a winner in $\mathbf{P}$.

To see how Borda violates Immunity to Spoilers, consider the following.

\begin{example}\label{BorderSpoiler} Let $\mathbf{P}_{-b}$ and $\mathbf{P}$ be the following profiles:
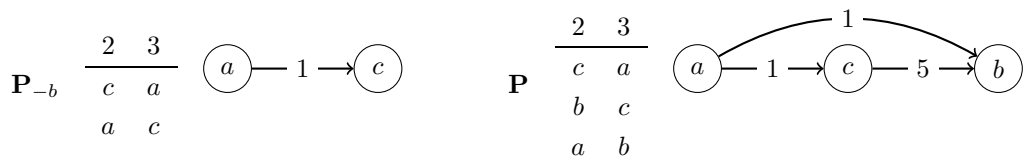
\begin{center}
$\mathbf{P}_{-b}$\quad
\begin{tabular}{cc}
$2$ & $3$    \\\hline
$c$ & $a$ \\
$a$ & $c$ \\
\end{tabular}\quad \begin{tikzpicture}

\node[circle,draw, minimum width=0.25in] at (0,0) (a) {$a$}; 
\node[circle,draw,minimum width=0.25in] at (2,0) (c) {$c$}; 
\path[->,draw,thick] (a) to node[fill=white] {$1$} (c);

\end{tikzpicture}\qquad\qquad $\mathbf{P}$\quad \begin{tabular}{cc}
$2$ & $3$   \\\hline
$c$ & $a$  \\
$b$ &  $c$  \\
$a$ &  $b$ \\
\end{tabular}\quad
\begin{tikzpicture}

\node[circle,draw, minimum width=0.25in] at (0,0) (a) {$a$}; 
\node[circle,draw,minimum width=0.25in] at (2,0) (c) {$c$}; 
\node[circle,draw,minimum width=0.25in] at (4,0) (b) {$b$}; 
\path[->,draw,thick] (a) to node[fill=white] {$1$} (c);
\path[->,draw,thick] (c) to node[fill=white] {$5$} (b);
\path[->,draw,thick,bend left] (a) to node[fill=white] {$1$} (b);
\end{tikzpicture}
\end{center}

\noindent According to Borda, $a$ defeats $c$ in $\mathbf{P}_{-b}$. Note that $a$ is majority preferred to $b$ in $\mathbf{P}$, and $b$ is defeated in $\mathbf{P}$ by both $a$ and $c$ according to Borda. But the addition of the loser $b$ spoils the election for $a$, as $c$ defeats $a$ in $\mathbf{P}$ according to Borda. This violates Immunity to Spoilers. It is also violates Independence of Clones (\citealt{Tideman1987}), as $b$ is a clone of $c$ (no candidates appear in between $b$ and $c$ on any voter's ballot). Finally, $c$ defeating $a$ in $\mathbf{P}$ but not in $\mathbf{P}_{-b}$ is another violation of Coherent IIA (see the proof of Proposition \ref{StrongStability}).\end{example}

In fact, Coherent IIA (together with the other mentioned axioms) implies an even stronger anti-spoiler axiom from \citealt{HP2020}:
\begin{itemize}
\item Stability for Winners (resp.~Strong Stability for Winners): if $a$ in undefeated in $\mathbf{P}_{-b}$, and $a$ is majority preferred to $b$ in $\mathbf{P}$ (resp.~$b$ is not majority preferred to $a$ in $\mathbf{P}$), then $a$ is still undefeated in $\mathbf{P}$.
\end{itemize}

\begin{proposition}\label{StrongStability} Anonymity, Neutrality,  Monotonicity (for two-candidate profiles), and Coherent IIA together imply Strong Stability for Winners.
\end{proposition}
\begin{proof} Suppose $a$ is undefeated in $\mathbf{P}_{-b}$ according to $f$ and that $b$ is not majority preferred to $a$ in $\mathbf{P}$. Suppose for contradiction that $a$ is defeated in $\mathbf{P}$ according to $f$. Since $b$ is not majority preferred to $a$ in $\mathbf{P}$, it follows by Lemma \ref{MajDefeatLem} that $b$ does not defeat $a$  in $\mathbf{P}$ according to $f$. Hence there is some $c\in X(\mathbf{P})\setminus\{b\}$ that defeats $a$ in $\mathbf{P}$ according to $f$. Then since $\mathbf{P}_{\mid \{a,c\}}=(\mathbf{P}_{-b})_{\mid \{a,c\}}$ and the margin graph of $\mathbf{P}_{-b}$ is obtained from that of $\mathbf{P}$ by deleting a candidate other than $a$ and $c$, it follows by Coherent IIA that $c$ defeats $a$ in $\mathbf{P}_{-b}$ according to $f$, contradicting our initial assumption. \end{proof}

\noindent Thus, contrary to \citealt{Maskin2020}, it is Coherent IIA rather than Modified IIA that mitigates spoiler effects.

\section{Characterization}\label{CharSection}

In this section, we prove our main result using the axioms proposed in Section \ref{AxiomsSection}. Given VCCRs $f$ and $g$, we say that $g$ is \textit{a refinement of} $f$ if for every profile $\mathbf{P}$ and $x,y\in X(\mathbf{P})$, if $x$ defeats $y$ in $\mathbf{P}$ according to $f$, then $x$ defeats $y$ in $\mathbf{P}$ according to $g$. It follows that the set of candidates selected by $\overline{g}$ (i.e., the set of undefeated candidates according to $g$) is always a \textit{subset} of the set of candidates selected by $\overline{f}$ (i.e., the set of undefeated candidates according to $f$). Thus, a refinement extends the defeat relation and may shrink the set of potential winners.

Given a class $\mathsf{C}$ of VCCRs and $g\in\mathsf{C}$, we say that $g$ is the \textit{most resolute} VCCR \textit{in} $\mathsf{C}$ if $g$ is a refinement of every VCCR in $\mathsf{C}$, or equivalently, if for every $f\in \mathsf{C}$, profile $\mathbf{P}$, and $x,y\in X(\mathbf{P})$, if $x$ defeats $y$ in $\mathbf{P}$ according to $f$, then $x$ defeats $y$ in $\mathbf{P}$ according to $g$. A number of voting procedures can be characterized as ``the most resolute procedure satisfying such and such properties'' (see, e.g., \citealt{Brandt2013,Brandt2014}). We now give such a characterization of Split Cycle.

\begin{theorem}\label{MainThm} Split Cycle is the most resolute of all VCCRs satisfying the six axioms for defeat:
\begin{enumerate}
\item[A1.] Anonymity and Neutrality: if $x$ defeats $y$ in $\mathbf{P}$, and $\mathbf{P}'$ is obtained from $\mathbf{P}$ by swapping the ballots assigned to two voters, then $x$ still defeats $y$ in $\mathbf{P}'$ (Anonymity); and if $x$ defeats $y$ in $\mathbf{P}$, and $\mathbf{P}'$ is obtained from $\mathbf{P}$ by swapping $x$ and $y$ on each voter's ballot, then $y$ defeats $x$ in $\mathbf{P}'$ (Neutrality).
\item[A2.]  Availability: for every $\mathbf{P}$, there is some undefeated candidate in $\mathbf{P}$.
\item[A3.] (Upward) Homogeneity: for every $\mathbf{P}$, if $x$ defeats $y$ in $\mathbf{P}$, then $x$ defeats $y$ in $2\mathbf{P}$.
\item[A4.] Monotonicity (for two-candidate profiles): if $x$ defeats $y$ in $\mathbf{P}$ (a two-candidate profile), and $\mathbf{P}'$ is obtained from $\mathbf{P}$ by some voter $i$ moving $x$ above the candidate that $i$ ranked immediately above $x$ in $\mathbf{P}$, then $x$ defeats $y$ in $\mathbf{P}'$.
\item[A5.]  Neutral Reversal: if $\mathbf{P}'$ is obtained  from $\mathbf{P}$ by adding two voters with reversed ballots, then $x$ defeats $y$ in $\mathbf{P}$ if and only if $x$ defeats $y$ in $\mathbf{P}'$.
\item[A6.]  Coherent IIA: if $x$ defeats $y$ in $\mathbf{P}$, and $\mathbf{P}'$ is a profile such that $\mathbf{P}_{\mid \{x,y\}}=\mathbf{P}'_{\mid \{x,y\}}$ and the margin graph of $\mathbf{P}'$ is obtained from that of $\mathbf{P}$ by deleting zero or more candidates other than $x$ and $y$ and deleting or reducing the margins on zero or more edges not connecting $x$ and $y$, then $x$ still defeats $y$ in $\mathbf{P}'$.
\end{enumerate}
\end{theorem}

\begin{proof} We have already observed that Split Cycle satisfies the axioms. Next, we show that for any VCCR $f$ satisfying the axioms for defeat and any profile $\mathbf{P}$, if $x$ defeats $y$ in $\mathbf{P}$ according to $f$, then $x$ defeats $y$ in $\mathbf{P}$ according to Split Cycle. Toward a contradiction, suppose $x$ defeats $y$ in $\mathbf{P}$ according to $f$ but not according to Split Cycle. Since $\mathbf{P}$ may have an odd number of voters, consider $2\mathbf{P}$. It follows by (Upward) Homogeneity that $x$ defeats $y$ in $2\mathbf{P}$ according to $f$. But according to Split Cycle, $x$ does not defeat $y$ in $2\mathbf{P}$ (whose qualitative margin graph is the same as $\mathbf{P}$). Since $x$ defeats $y$ according to $f$, by Majority Defeat (Lemma \ref{MajDefeatLem}), we have $Margin_{2\mathbf{P}}(x,y)>0$. Moreover, since $2\mathbf{P}$ has an even number of voters, $Margin_{2\mathbf{P}}(x,y)$ is even. Now since $x$ does not defeat $y$ in  $2\mathbf{P}$ according to Split Cycle, by Lemma \ref{OnlySomeCycles} there is a majority cycle $x\to y\to z_1\to\dots\to z_n\to x$ in $2\mathbf{P}$ such that $Margin_{2\mathbf{P}}(x,y)$ is less than or equal to every margin along the cycle. Let $\mathbf{Q}$ be obtained from $\mathbf{2P}$ by adding zero or more reversal pairs of voters so that $|V(\mathbf{Q})|\geq (n+2)Margin_{2\mathbf{P}}(x,y)$. Then by (Upward) Neutral Reversal, $x$ defeats $y$ in $\mathbf{Q}$ according to $f$. Let $\mathcal{M}'$ be the weighted directed graph obtained from $\mathcal{M}(2\mathbf{P})$ by deleting  all candidates except $x,y,z_1,\dots,z_n$ and all edges except the edges in the cycle and reducing the weights on all remaining edges so they are equal to $Margin_{2\mathbf{P}}(x,y)$. Let $k=Margin_{2\mathbf{P}}(x,y)/2$. We now construct a profile $\mathbf{P}'$ whose margin graph is $\mathcal{M}'$.\footnote{This standard kind of construction is used in the proof of McGarvey's theorem (\citealt{McGarvey1953}).} For each edge $a\to b$ in $x\to y\to z_1\to\dots\to z_n\to x$,  where $c_1,\dots,c_n$ are the elements from $\{x,y,z_1,\dots,z_n\}\setminus\{a,b\}$ such that $b\to c_1\to\dots\to c_n\to a$, we add voters to $\mathbf{P}'$ as follows:
\begin{itemize}
\item if $a=x$ (and hence $b=y$), then for $k$ voters from $V(\mathbf{Q})$ who rank $x$ above $y$ in $\mathbf{Q}$, add them to $\mathbf{P}'$ with the ballot $\boldsymbol{a},\boldsymbol{b},c_1,\dots,c_n$ (which in this case is $\boldsymbol{x}, \boldsymbol{y},z_1,\dots,z_n$), and for  $k$ other voters from $V(\mathbf{Q})$ who rank $x$ above $y$ in $\mathbf{Q}$, add them to $\mathbf{P}'$ with the ballot $c_n,\dots,c_1, \boldsymbol{a},\boldsymbol{b}$ (which in this case is $z_1,\dots,z_n, \boldsymbol{x}, \boldsymbol{y}$). See the first and second columns in the profile in Figure \ref{McGarvey}.
\item if $a\neq x$ and $y$ occurs before $x$ in the sequence $b,c_1,\dots,c_n,a$, then for $k$ voters from $V(\mathbf{Q})$ who rank $y$ above $x$ in $\mathbf{Q}$, add them to $\mathbf{P}'$ with the ballot $\boldsymbol{a},\boldsymbol{b},c_1,\dots,c_n$, and for  $k$ voters from $V(\mathbf{Q})$ who rank $x$ above $y$ in $\mathbf{Q}$, add them to $\mathbf{P}'$ with the ballot $c_n,\dots,c_1, \boldsymbol{a},\boldsymbol{b}$. See, e.g., the third and fourth columns in the profile in Figure \ref{McGarvey}.
\item if $a\neq x$ and $x$ occurs before $y$ in the sequence $b,c_1,\dots,c_n,a$, then for $k$ voters from $V(\mathbf{Q})$ who rank $x$ above $y$ in $\mathbf{Q}$, add them to $\mathbf{P}'$ with the ballot $\boldsymbol{a},\boldsymbol{b},c_1,\dots,c_n$, and for  $k$ voters from $V(\mathbf{Q})$ who rank $y$ above $x$ in $\mathbf{Q}$, add them to $\mathbf{P}'$ with the ballot $c_n,\dots,c_1, \boldsymbol{a},\boldsymbol{b}$. See, e.g., the fifth and sixth columns in the profile in Figure \ref{McGarvey}.
\end{itemize}
This construction uses $2k+(n+1)k$ voters from $V(\mathbf{Q})$ who rank $x$ over $y$ and $(n+1)k$ voters from $V(\mathbf{Q})$ who rank $y$ over $x$, for a total of $(n+2)Margin_{2\mathbf{P}}(x,y)$ voters from $V(\mathbf{Q})$. Then $\mathbf{P}'$ has the form in Figure~\ref{McGarvey}. Observe that $\mathcal{M}(\mathbf{P}')=\mathcal{M}'$ (e.g., $Margin_{\mathbf{P}'}(x,y)=2k=Margin_{2\mathbf{P}}(x,y)$, $Margin_{\mathbf{P}'}(x,z_1)=0$, etc.).
\begin{figure}[h]
  \begin{center}
\begin{tabular}{cc|cc|cc|c|cc}
 $k$ &  $k$ & $k$ & $k$ & $k$ & $k$ & $\cdots$ & $k$ & $k$    \\\hline
$\boldsymbol{x}$ & $z_n$ & $\boldsymbol{y}$ & $x$ & $\boldsymbol{z_1}$ & $y$ & $\cdots$ & $ \boldsymbol{z_n}$ & $z_{n-1}$      \\
$\boldsymbol{y}$ &  $\vdots$ & $\boldsymbol{z_1}$ & $z_n$ &  $\boldsymbol{z_2}$ & $x$ & $\cdots$ & $\boldsymbol{x}$ & $\vdots$   \\
$z_1$ &  $\vdots$ & $z_2$ & $\vdots$ & $\vdots$ & $z_n$ & $\cdots$ & $y$ & $z_1$    \\
$\vdots$ &  $z_1$ & $\vdots$ & $z_2$ & $z_n$ & $\vdots$ & $\cdots$ & $z_1$ & $y$   \\
$\vdots$ &  $\boldsymbol{x}$ & $z_n$ & $\boldsymbol{y}$ & $x$ & $\boldsymbol{z_1}$ & $\cdots$ & $\vdots$ & $ \boldsymbol{z_n}$   \\
$z_n$ &  $\boldsymbol{y}$ & $x$ & $\boldsymbol{z_1}$ & $y$ & $\boldsymbol{z_2}$ & $\cdots$ & $z_{n-1}$ & $\boldsymbol{x}$   \\
\end{tabular}
\end{center}
\caption{the profile $\mathbf{P}'$.} \label{McGarvey}
\end{figure}

 Now we claim that $x$ does not defeat $y$ in $\mathbf{P}'$ according to $f$. Toward a contradiction, suppose $x$ does defeat $y$ in $\mathbf{P}'$ according to $f$. Let $\sigma(\mathbf{P}')$ be the profile obtained from $\mathbf{P}'$ by the permutation $\sigma$ that maps each $a\in \{x,y,z_1,\dots,z_n\}$ to the unique $b\in  \{x,y,z_1,\dots,z_n\}$ such that $a\to b$ in $\mathbf{P}'$, as shown in Figure \ref{McGarveyPermuted}.
\begin{figure}[h]
  \begin{center}
\begin{tabular}{cc|cc|cc|c|cc}
  $k$ & $k$ & $k$ & $k$ & $k$ & $k$ & $\cdots$  &  $k$ &  $k$    \\\hline
  $\boldsymbol{y}$ & $x$ & $\boldsymbol{z_1}$ & $y$ & $\boldsymbol{z}_2$  & $z_1$ & $\cdots$ & $\boldsymbol{x}$ & $z_n$      \\
  $\boldsymbol{z_1}$ & $z_n$ &  $\boldsymbol{z_2}$ & $x$ & $\boldsymbol{z}_3$ & $y$ & $\cdots$ & $\boldsymbol{y}$ &  $\vdots$   \\
 $z_2$ & $\vdots$ & $\vdots$ & $z_n$ & $\vdots$ & $x$  & $\cdots$  & $z_1$ &  $\vdots$    \\
  $\vdots$ & $z_2$ & $z_n$ & $\vdots$ & $z_n$ & $\vdots$  & $\cdots$  & $\vdots$ &  $z_1$  \\
  $z_n$ & $\boldsymbol{y}$ & $x$ & $\boldsymbol{z_1}$ & $y$ & $\boldsymbol{z}_2$ & $\cdots$  & $\vdots$ &  $\boldsymbol{x}$   \\
  $x$ & $\boldsymbol{z_1}$ & $y$ & $\boldsymbol{z_2}$ & $z_1$ & $\boldsymbol{z}_3$  & $\cdots$ & $z_n$ &  $\boldsymbol{y}$   \\
\end{tabular}
\end{center} 
\caption{the profile $\sigma(\mathbf{P}')$.} \label{McGarveyPermuted}
\end{figure}

\noindent By Neutrality,\footnote{Here we use the permutation version of Neutrality in Footnote \ref{Permutation}.} since $x$ defeats $y$ in $\mathbf{P}'$, it follows that $y$ defeats $z_1$ in $\sigma(\mathbf{P}')$. But $\mathbf{P}'$ can  obviously be obtained from $\sigma(\mathbf{P}')$ by a permutation of the voters (e.g., in Figures \ref{McGarvey}-\ref{McGarveyPermuted}, the first column in $\mathbf{P}'$ is the same as the second to last column in $\sigma(\mathbf{P}')$). Thus, by Anonymity, since $y$ defeats $z_1$ in $\sigma(\mathbf{P}')$, it follows that $y$ defeats $z_1$ in $\mathbf{P}'$. By similar reasoning using $\sigma(\sigma(\mathbf{P}'))$, we have that $z_1$ defeats $z_2$ in $\mathbf{P}'$, etc., until we conclude that $xDyDz_1D\dots Dz_nDx$ where $D$ is the defeat relation in $\mathbf{P}'$.  This contradicts Availability. Hence $x$ does not defeat $y$ in $\mathbf{P}'$ according to~$f$.  

Since $V(\mathbf{P}')\subseteq V(\mathbf{Q})$ and  $Margin_\mathbf{Q}(x,y)=Margin_{\mathbf{P}'}(x,y)$, it follows that half of the voters in ${V(\mathbf{Q})\setminus V(\mathbf{P}')}$ (which may be empty) rank $x$ above $y$ and half of the voters in $V(\mathbf{Q})\setminus V(\mathbf{P}')$ rank $y$ above $x$. Let $\mathbf{P}''$ be obtained from $\mathbf{P}'$ as follows: for each voter in $V(\mathbf{Q})\setminus V(\mathbf{P}')$ who ranks $x$ above $y$, add them to $\mathbf{P}''$ with the ballot $x,y,z_1,\dots, z_n$, and for each voter $V(\mathbf{Q})\setminus V(\mathbf{P}')$ who ranks $y$ above $x$, add them to $\mathbf{P}''$ with the ballot $z_n,\dots,z_1, y,x$. Thus, $\mathbf{P}''$ is obtained from $\mathbf{P}'$ by adding zero or more reversal pairs of voters, so by (Downward) Neutral Reversal, since $x$ does not defeat $y$ in $\mathbf{P}'$ according to $f$, it follows that $x$ does not defeat $y$ in $\mathbf{P}''$ according to $f$. Finally, we have:
\begin{itemize}
\item $V(\mathbf{Q})=V(\mathbf{P}'')$ and $\mathbf{Q}_{\mid \{x,y\}}=\mathbf{P}''_{\mid \{x,y\}}$;
\item $\mathcal{M}(\mathbf{P}'')$ is obtained from $\mathcal{M}(\mathbf{Q})$ by deleting zero or more candidates other than $x$ and $y$ and deleting or reducing the margins on zero or more edges other than the $x\to y$ edge.
\end{itemize}
Thus, by Coherent IIA, since $x$ defeats $y$ in $\mathbf{Q}$ according to $f$, we have that $x$ defeats $y$ in $\mathbf{P}''$ according to $f$, contradicting what we derived above.\end{proof}

Of course Split Cycle is not the only VCCR that satisfies the six axioms for defeat. For example, the null VCCR according to which no one ever defeats anyone else satisfies all six axioms---although it can easily be ruled out by other axioms that Split Cycle satisfies, such as the Pareto axiom. Another example is the VCCR according to which $x$ defeats $y$ if $Margin_\mathbf{P}(x,y)$ is greater than the splitting number of every majority cycle in $\mathbf{P}$, not only those majority cycles containing $x$ and $y$. Since these VCCRs are not refinements of Split Cycle,\footnote{E.g., in the profile $\mathbf{P}$ in Example \ref{WinByEx}, neither of the mentioned VCCRs judges that $a$ defeats $d$.} they do not satisfy the seventh ``axiom'' that the VCCR should be the most resolute VCCR among those satisfying the first six axioms. A natural next step would be to obtain another axiomatic characterization of Split Cycle as the only VCCR satisfying some axioms without reference to resoluteness.

Theorem \ref{MainThm} shows that using a VCCR other than Split Cycle requires either violating one of the six axioms for defeat or sacrificing resoluteness. For these and other reasons (see \citealt{HP2020}), we settle on Split Cycle as our preferred VCCR and hence as our preferred answer to the question of when one candidate should defeat another in a democratic election using ranked ballots.

\begin{remark}\label{Tiebreaking} When there are multiple undefeated candidates, but a single winner must be chosen, some further tiebreaking process must select the ultimate winner from  the undefeated candidates. However, we need not interpret that process as establishing additional relations of \textit{defeat} between candidates in the politically significant sense of defeat. This is obvious if we randomly choose the ultimate winner from the undefeated candidates. But the point may also apply if we first make a deterministic choice of a subset of undefeated candidates before resorting to random choice if necessary.\end{remark}

\section{Escaping impossibility}\label{EscapeSection}

In this section, we address the question: how does Split Cycle escape Arrow's Impossibility Theorem and related impossibility results? That is, how did we relax Arrow's assumptions in order to avoid the existence of a dictator, vetoers, etc.? In Section \ref{FixedImposs}, we recall the standard formulation of Arrow's theorem and related results, and we explain how Split Cycle escapes these results. In Section \ref{IIAVCCR}, we reformulate these results in the variable-candidate setting in which we characterized Split Cycle. Finally, in Section \ref{AlphaVoting}, we consider some simple impossibility results based not on IIA but instead on a choice-consistency principle sometimes conflated with IIA, allegedly even by Arrow himself (see Appendix \ref{ArrowSection}).

\subsection{Impossibility theorems in the fixed-candidate setting}\label{FixedImposs}

\subsubsection{Arrow's Theorem}

Arrow \citeyearpar{Arrow1963} worked in a fixed-voter and fixed-candidate setting (see \citealt{Campbell2002} and \citealt{Penn2015} for modern presentations). Fix nonempty finite sets $V$ and $X$ of voters and candidates, respectively. A $(V,X)$-profile is a profile $\mathbf{P}$ as in Definition \ref{ProfileDef} in which $V(\mathbf{P})=V$ and $X(\mathbf{P})=X$. A \textit{$(V,X)$-collective choice rule} (or $(V,X)$-CCR) is a function on the set of $(V,X)$-profiles such that for any $(V,X)$-profile $\mathbf{P}$, $f(\mathbf{P})$ is an asymmetric binary relation on $X$.\footnote{Note that we have built the condition of Universal Domain with respect to $(V,X)$ into the definition of a $(V,X)$-CCR.} If $f$ is a VCCR as in Section \ref{Prelim}, then for any finite $V\subset\mathcal{V}$ and finite $X\subset\mathcal{X}$, the restriction $f_{\mid V,X}$ of $f$ to the set of $(V,X)$-profiles is a $(V,X)$-CCR. In particular, for any such $(V,X)$, the Split Cycle VCCR restricts to the Split Cycle $(V,X)$-CCR. Our question is: how does the Split Cycle $(V,X)$-CCR escape Arrow's Theorem?

Let $f$ be a $(V,X)$-CCR. To state Arrow's Theorem, we recall the following key notions:
\begin{itemize}
\item $f$ is a $(V,X)$-\textit{social welfare function} (or $(V,X)$-SWF) if for any $(V,X)$-profile $\mathbf{P}$, $f(\mathbf{P})$ is a strict weak order (recall Section \ref{Prelim});
\item $f$ satisfies \textit{Independence of Irrelevant Alternatives} (IIA) if for any $(V,X)$-profiles $\mathbf{P}$ and $\mathbf{P}'$ and $x,y\in X$, if $\mathbf{P}_{\mid\{x,y\}}=\mathbf{P}'_{\mid\{x,y\}}$, then $x$ defeats $y$ in $\mathbf{P}$ according to $f$ if and only if $x$ defeats $y$ in $\mathbf{P}'$ according to~$f$;
\item $f$ satisfies \textit{Pareto} if for any $(V,X)$-profile $\mathbf{P}$ and $x,y\in X$, if $x\mathbf{P}_iy$ for all $i\in V$, then $x$ defeats $y$ in $\mathbf{P}$ according to $f$;
\item an $i\in V$ is a \textit{dictator for $f$} if for all $(V,X)$-profiles $\mathbf{P}$ and $x,y\in X$, if $x\mathbf{P}_iy$, then  $x$ defeats $y$ in $\mathbf{P}$ according to $f$.
\end{itemize}

Then Arrow's famous Impossibility Theorem can be stated as follows.\footnote{\label{LinearProfiles}Arrow considered profiles where each voter's strict preference relation is a strict weak order, whereas we have assumed strict linear orders. However, it is well known that Arrow's Theorem can be proved for strict linear order profiles (see, e.g., \citealt[p.~208]{Fishburn1973}). In fact, Arrow's Theorem for strict linear order profiles is a corollary of the statement of Arrow's Theorem for strict weak orders, by applying Lemma 3.4 of \citealt{HP2020a}.}

\begin{theorem}[\citealt{Arrow1963}]\label{ArrowThm} Assume $|X|\geq 3$. Any $(V,X)$-SWF satisfying IIA and Pareto has a dictator.
\end{theorem}
\begin{remark}\label{StrongDictator} Since our profiles are profiles of linear ballots, the conclusion of Arrow's theorem can be strengthened to say that $f$ has a \textit{strong dictator}, i.e., an $i\in V$ such that for all $(V,X)$-profiles $\mathbf{P}$ and $x,y\in X$, we have that  $x$ defeats $y$ in $\mathbf{P}$ according to $f$ \textit{if and only if} $x\mathbf{P}_iy$, i.e., $f(\mathbf{P})=\mathbf{P}_i$.
\end{remark}

The Split Cycle $(V,X)$-CCR avoids Arrow's theorem due to the following facts:
\begin{enumerate}
\item We weaken IIA to Coherent IIA.\footnote{Strictly speaking, we have stated Coherent IIA in a variable-candidate setting and IIA in a fixed-candidate setting, so they are not comparable in strength, but see  Proposition \ref{VIIACoherent} in Section \ref{IIAVCCR}.}
\item We weaken Arrow's assumption that the defeat relation is a strict weak order to it being acyclic.
\end{enumerate}
\textit{Neither of these moves by itself} is sufficient to escape Arrow-style impossibility theorems, as we show below.

\subsubsection{Baigent's Theorem}\label{BaigentSection}

To see that weakening IIA to Coherent IIA is not sufficient, we first observe that Coherent IIA implies a weakening of IIA known as Weak IIA, which states that if $\mathbf{P}_{\mid\{x,y\}}=\mathbf{P}'_{\mid\{x,y\}}$ and $x$ defeats $y$ in $\mathbf{P}$ according to $f$, then $y$ does \textit{not} defeat $x$ in $\mathbf{P}'$ according to $f$.

\begin{lemma}\label{CoherentToWeak} If $f$ is a VCCR satisfying Coherent IIA, then for any finite $V\subset \mathcal{V}$ and finite $X\subset\mathcal{X}$, $f_{\mid V,X}$ satisfies Weak IIA.
\end{lemma}
\begin{proof} Suppose that $\mathbf{P}_{\mid \{x,y\}}=\mathbf{P}'_{\mid \{x,y\}}$ and $x$ defeats $y$ in $\mathbf{P}$. Then by Coherent IIA, $x$ defeats $y$ in $\mathbf{P}_{\mid \{x,y\}}$, so $y$ does not defeat $x$ in $\mathbf{P}_{\mid \{x,y\}}$. Now if $y$ defeats $x$ in $\mathbf{P}'$, then by Coherent IIA, $y$ defeats $x$ in $\mathbf{P}'_{\mid \{x,y\}}$ and hence in $\mathbf{P}_{\mid \{x,y\}}$, since $\mathbf{P}_{\mid \{x,y\}}=\mathbf{P}'_{\mid \{x,y\}}$, which is a contradiction. Therefore, $y$ does not defeat $x$ in $\mathbf{P}'$.\end{proof}

Under Weak IIA, Baigent \citeyearpar{Baigent1987} proved an Arrow-style impossibility theorem asserting the existence of a vetoer instead of a dictator.\footnote{Cf.~Campbell and Kelly \citeyearpar{Campbell2000}, who observe that at least four candidates are required for Baigent's result.} Given a $(V,X)$-CCR $f$, a voter $i\in V$ is a \textit{vetoer for $f$} if for all $(V,X)$-profiles $\mathbf{P}$ and $x,y\in X(\mathbf{P})$, if $x\mathbf{P}_iy$, then  $y$ does not defeat $x$ in $\mathbf{P}$ according to~$f$. 

\begin{theorem}[\citealt{Baigent1987}]\label{BaigentThm} Assume $|X|\geq 4$. Any $(V,X)$-SWF satisfying Weak IIA and Pareto has a vetoer.
\end{theorem}
\noindent The existence of a vetoer for an SWF is inconsistent with the SWF satisfying both Pareto and Anonymity.

\begin{proposition}\label{AnonProp} Suppose $|V|\geq 2$ and $|X|\geq 3$. Let $f$ be a $(V,X)$-SWF satisfying Pareto. If $f$ has a vetoer, then $f$ has a unique vetoer and hence violates Anonymity.
\end{proposition}
\begin{proof} Suppose there are two vetoers $j$ and $k$. Consider a profile $\mathbf{P}$ in which (i) $x\mathbf{P}_iy$ for all $i\in V$, (ii) $y\mathbf{P}_jz$, and (iii) $z\mathbf{P}_kx$. By (i) and Pareto, $x$ defeats $y$ in $\mathbf{P}$ according to $f$. By (ii), $z$ does not defeat $y$, since $j$ is a vetoer. By (iii), $x$ does not defeat $z$, since $k$ is a vetoer. But since $f$ is an SWF, $f(\mathbf{P})$ is a strict weak order, so if $x$ defeats $y$, then either $z$ defeats $y$ or $x$ defeats $z$. Thus, we have a contradiction.\end{proof}

As a corollary of Theorem \ref{BaigentThm} and Proposition \ref{AnonProp}, we have the following.

\begin{corollary}\label{BaigentCor} Assume  $|X|\geq 4$. There is no $(V,X)$-SWF satisfying Weak IIA, Pareto, and Anonymity.
\end{corollary}

\noindent In light of Lemma \ref{CoherentToWeak}, Theorem \ref{BaigentThm}, and Corollary \ref{BaigentCor}, only weakening IIA to Coherent IIA is not sufficient to escape Arrow-style impossibility results.

\subsubsection{Blau-Deb Theorem}

Weakening Arrow's assumption that the defeat relation is a strict weak order to it being acyclic is also not sufficient by itself. Blau and Deb \citeyearpar{Blau1977} prove a vetoer theorem for acyclic CCRs under IIA together with Neutrality and Monotonicity (recall Section \ref{StandardAx}). Let $f$ be a $(V,X)$-CCR.  A coalition $C\subseteq V$ of voters has \textit{veto power for $f$} if for any $(V,X)$-profile $\mathbf{P}$ and $x,y\in X$, if $x\mathbf{P}_iy$ for all $i\in C$, then $y$ does not defeat $x$ in $\mathbf{P}$ according to $f$. 

\begin{theorem}[\citealt{Blau1977}]\label{BlauDebThm} Let $f$ be an acyclic $(V,X)$-CCR satisfying IIA, Neutrality, and Monotonicity.
\begin{enumerate}
\item\label{BlauDebThma} For any partition of $V$ into at most $|X|$-many coalitions, at least one of the coalitions has veto power.
\item\label{BlauDebThmb} If $|X|\geq |V|$, then $f$ has a vetoer.
\end{enumerate}
\end{theorem}
\begin{remark} Part \ref{BlauDebThmb} is an immediate consequence of part \ref{BlauDebThma} by considering the finest partition.
\end{remark}
\begin{remark}\label{Availability} Inspection of the proof of the Veto Theorem in \citealt{Blau1977} shows that the assumption of acyclicity may be replaced by the weaker axiom of Availability (recall Section \ref{StandardAx}).
\end{remark}

As an example of applying Theorem \ref{BlauDebThm}.\ref{BlauDebThma}, if there are five candidates, then for any partition of the electorate into five coalitions---say, five coalitions of equal size---one of the five coalitions has veto power (and hence, assuming Anonymity, all coalitions of the same size would have veto power). Moreover, in the variable-candidate setting, we can use Theorem \ref{BlauDebThm} to prove the existence of a single vetoer under a variable-candidate version of IIA (see Proposition \ref{BlauDebThm2} below), without the assumption that $|X|\geq |V|$. Thus, weakening Arrow's strict weak order assumption to the assumption of acyclicity (or even Availability) is not enough by itself to escape Arrow-style impossibility theorems, if we would like to retain Neutrality and Monotonicity.

It is the combination of weakening IIA to Coherent IIA and weakening Arrow's strict weak order assumption to acyclicity that allows Split Cycle to escape Arrow-style impossibility theorems. 

\subsection{Impossibility theorems in the variable-candidate setting}\label{IIAVCCR}

Since we have analyzed Split Cycle as a VCCR in this paper, to properly make claims about how Split Cycle relates to Arrow-style impossibility theorems, we should recast these results in the variable-election setting. In this setting, there are two versions of Arrow's Independence of Irrelevant Alternatives (IIA).

\begin{definition}\label{IIAdef} Let $f$ be a VCCR.
\begin{enumerate}
\item\label{IIAdef1} $f$ satisfies \textit{fixed-candidate Independence of Irrelevant Alternatives} (FIIA) if for any profiles $\mathbf{P}$ and $\mathbf{P}'$ with $X(\mathbf{P})=X(\mathbf{P}')$, if $\mathbf{P}_{\mid\{x,y\}}=\mathbf{P}'_{\mid\{x,y\}}$, then  $x$ defeats $y$ in $\mathbf{P}$ according to $f$ if and only if $x$ defeats $y$ in $\mathbf{P}'$ according to $f$; 
\item $f$ satisfies \textit{variable-candidate Independence of Irrelevant Alternatives} (VIIA) if for any profiles $\mathbf{P}$ and $\mathbf{P}'$ with $x,y\in X(\mathbf{P})\cap X(\mathbf{P}')$, if $\mathbf{P}_{\mid\{x,y\}}=\mathbf{P}'_{\mid\{x,y\}}$, then $x$ defeats $y$ in $\mathbf{P}$ according to $f$ if and only if $x$ defeats $y$ in $\mathbf{P}'$ according to $f$.
\end{enumerate}
\end{definition}

We suggest in Appendix \ref{ArrowSection} that if asked to formulate his axioms for VCCRs, Arrow would formulate IIA as~VIIA. Our Coherent IIA is a weakening of VIIA, as Coherent IIA strengthens the assumption from $\mathbf{P}_{\mid\{x,y\}}=\mathbf{P}'_{\mid\{x,y\}}$ to the assumption that not only $\mathbf{P}_{\mid\{x,y\}}=\mathbf{P}'_{\mid\{x,y\}}$ but also that the margin graph of $\mathbf{P}'$ is obtained from that of $\mathbf{P}$ in a certain way. 

\begin{proposition}\label{VIIACoherent} Any VCCR satisfying VIIA also satisfies Coherent IIA.
\end{proposition}
\noindent We reject VIIA in favor of Coherent IIA for the reasons explained in Section \ref{CoherentIIASection}.

 Arrow's Impossibility Theorem can be stated in the variable-election setting using some additional notions. First, given a finite $V\subset \mathcal{V}$, a \textit{$V$-profile} is a profile $\mathbf{P}$ as in Definition \ref{ProfileDef} in which $V(\mathbf{P})=V$. Second, given $i\in V$ and finite $X\subset \mathcal{X}$, we say that $i$ is \textit{$(V,X)$-dictator} (resp.~\textit{$V$-dictator}) \textit{for $f$} if for any $(V,X)$-profile (resp.~$V$-profile) $\mathbf{P}$ and $x,y\in X(\mathbf{P})$, $x\mathbf{P}_iy$ implies that $x$ defeats $y$ in $\mathbf{P}$ according to $f$.
 
\begin{theorem}[Arrow's Theorem for VSWFs]\label{Arrow} Suppose $f$ is a VSWF satisfying the Pareto principle.
\begin{enumerate}
\item\label{Arrow1} If $f$ satisfies FIIA, then for any finite sets $V\subset\mathcal{V}$ and $X\subset\mathcal{X}$  with $|X|\geq 3$, there is a $(V,X)$-dictator for $f$.
\item\label{Arrow2} If $f$ satisfies VIIA, then for any finite set $V\subset\mathcal{V}$, there is a $V$-dictator for $f$.
\end{enumerate}
\end{theorem}
\begin{proof} For part \ref{Arrow1}, let $f_{\mid V,X}$ be the restriction of $f$ to $(V,X)$-profiles. Then $f_{\mid V,X}$ is a $(V,X)$-SWF as in Section \ref{EscapeSection} satisfying IIA and Pareto. Since $|X|\geq 3$, Arrow's Theorem (Theorem \ref{ArrowThm}) gives us the desired $(V,X)$-dictator for $f$.

For part \ref{Arrow2}, consider any finite $V\subset\mathcal{V}$. Pick some finite $X_0\subset \mathcal{X}$ such that $|X_0|\geq 3$. Then as in part \ref{Arrow1}, Arrow's Theorem applied to $f_{\mid V,X_0}$ gives us an $i_{V,X_0}\in V$ who is a $(V,X_0)$-dictator for $f$. We claim that $i_{V,X_0}$ is a $V$-dictator for $f$. Let $\mathbf{Q}$ be a $V$-profile. We must show that for all $x,y\in X(\mathbf{Q})$, $x\mathbf{Q}_{i_{V,X_0}}y$ implies that $x$ defeats $y$ in $\mathbf{Q}$ according to $f$. Suppose $x\mathbf{Q}_{i_{V,X_0}}y$. Let $X=X_0\cup X(\mathbf{Q})$. Since $|X|\geq 3$, Arrow's Theorem applied to $f_{\mid V,X}$ gives us an $i_{V,X}\in V$ who is a $(V,X)$-dictator for $f$. We claim that $i_{V,X}=i_{V,X_0}$. Suppose not. There is a $(V,X)$-profile $\mathbf{P}^\star$ such that for some $a,b\in X_0$, voter $i_{V,X_0}$ ranks $a$ above $b$ in $\mathbf{P}^\star$ while $i_{V,X}$ ranks $b$ above $a$ in $\mathbf{P}^\star$. Then since $i_{V,X}$ is a $(V,X)$-dictator, $b$ defeats $a$ in $\mathbf{P}^\star$ according to $f$. Hence by VIIA, $b$ defeats $a$ in $\mathbf{P}^\star_{\mid X_0}$ according to $f$. But this contradicts the fact that $i_{V,X_0}$ is a $(V,X_0)$-dictator, given that $X(\mathbf{P}^\star_{\mid X_0})=X_0$ and $i_{V,X_0}$ ranks $a$ above $b$ in $\mathbf{P}^\star_{\mid X_0}$. Hence $i_{V,X}=i_{V,X_0}$, so $i_{V,X_0}$ is a $(V,X)$-dictator for $f$.  Now let $\mathbf{Q}^+$ be any $(V,X)$-profile such that $\mathbf{Q}^+_{\mid X}=\mathbf{Q}$. Then $x\mathbf{Q}_{i_{V,X_0}}y$ implies $x\mathbf{Q}^+_{i_{V,X_0}}y$, so $x$ defeats $y$ in $\mathbf{Q}^+$ according to $f$ because $i_{V,X_0}$ is a $(V,X)$-dictator. Hence by VIIA, $x$ defeats $y$ in $\mathbf{Q}$ according to $f$, which completes the proof, as diagrammed in Figure \ref{UpDown} with $Y=X(\mathbf{Q})$.\end{proof}

\begin{figure}[h]
\begin{center}
\begin{minipage}{2in}\begin{tikzpicture}

\node at (0,0) (a) {$(V,X_0)$-dictator}; 
\node[minimum width=0.25in] at (4,0) (b) {$(V,Y)$-dictator}; 
\node[minimum width=0.25in] at (2,1.5) (c) {$(V,X_0\cup Y)$-dictator}; 

\path[->,draw,thick] (a) to   (b);
\path[->,draw,thick,dashed] (c) to  (b);
\path[->,draw,thick, dashed] (a) to (c);

\end{tikzpicture}
\end{minipage}\end{center}
\caption{To show that any $(V,X_0)$-dictator is also a $(V,Y)$-dictator, we first show that any $(V,X_0)$-dictator is also a $(V,X_0\cup Y)$-dictator and then show that any $(V,X_0\cup Y)$-dictator is also a $(V,Y)$-dictator.}\label{UpDown}
\end{figure}

\begin{remark} As in Remark \ref{StrongDictator}, since our profiles are profiles of linear ballots, the conclusions of parts \ref{Arrow1} and \ref{Arrow2} of Theorem \ref{Arrow} can be strengthened with `strong dictator' in place of `dictator'.
\end{remark}

\begin{remark} There are VSWFs satisfying Pareto and VIIA for which there is no $i\in \mathcal{V}$ who is a $V$-dictator with respect to all finite $V\subset \mathcal{V}$ with $i\in V$. For example, let $\mathcal{V}$ be the set of natural numbers, and for any profile $\mathbf{P}$, let $f(\mathbf{P})=\mathbf{P}_{\mathrm{max}(V(\mathbf{P}))}$, where $\mathrm{max}(V(\mathbf{P}))$ is the greatest number in the set $V(\mathbf{P})$. Thus, in the variable-voter setting Arrow's axioms are consistent with different electorates having different dictators.\end{remark}

Just as Arrow's Theorem can be adapted to the variable-election setting, so  can Baigent's Theorem (Theorem \ref{BaigentThm}), which we leave as an exercise to the reader (hint: use Proposition \ref{AnonProp} to obtain the analogue of Theorem \ref{Arrow}.\ref{Arrow2}). More interesting is the reformulation of the Blau-Deb Theorem (Theorem \ref{BlauDebThm}) in the variable-election setting---in particular, the variable-candidate setting---as VIIA allow us to strengthen the conclusion of the theorem to state the existence of a vetoer without the restriction that $|X|\geq |V|$.  

To state the variable-candidate version of the Blau-Deb Theorem, we need the following notions. Given finite $V\subset \mathcal{V}$, $i\in V$, finite $X\subset \mathcal{X}$, and $a,b\in X$, we say that:

\begin{itemize}
\item $i$ is a \textit{$(V,X)$-vetoer for $f$ on $(a,b)$} if for all $(V,X)$-profiles $\mathbf{P}$,  if $a\mathbf{P}_ib$, then  $b$ does not defeat $a$ in $\mathbf{P}$ according to~$f$;
\item $i$ is a \textit{$(V,X)$-vetoer for $f$} if for every $a,b\in X$, $i$ is a $(V,X)$-vetoer for $f$ on $(a,b)$;
\item $i$ is a \textit{$V$-vetoer for $f$} if for every finite $X\subset \mathcal{X}$, $i$ is a $(V,X)$-vetoer for $f$.
\end{itemize}

\begin{theorem}\label{BlauDebThm2} If $f$ is a VCCR satisfying VIIA, Availability, Neutrality, and Monotonicity, then for any finite $V\subset \mathcal{V}$, $f$ has a $V$-vetoer.
\end{theorem}
\begin{proof} Consider any finite $V\subset \mathcal{V}$. Pick some finite $X_0\subset \mathcal{X}$ such that $|X_0|\geq |V|$. Then Theorem \ref{BlauDebThm}.\ref{BlauDebThmb} (and Remark \ref{Availability}) applied to $f_{\mid V,X_0}$ gives us an $i_{V,X_0}\in V$ who is a $(V,X_0)$-vetoer for $f$. We claim that $i_{V,X_0}$ is a $V$-vetoer for $f$. Let $\mathbf{Q}$ be a $V$-profile. We must show that for all $x,y\in X(\mathbf{Q})$, $x\mathbf{Q}_{i_{V,X_0}}y$ implies that $y$ does not defeat $x$ in $\mathbf{Q}$ according to $f$. Suppose $x\mathbf{Q}_{i_{V,X_0}}y$. Let $X=X_0\cup X(\mathbf{Q})$. We claim that for any $a,b\in X_0$, voter $i_{V,X_0}$ is a $(V,X)$-vetoer on $(a,b)$. Suppose $\mathbf{P}$ is a $(V,X)$-profile such that $a\mathbf{P}_{i_{V,X_0}}b$. Then  $i_{V,X_0}$ ranks $a$ above $b$ in the restricted profile  $\mathbf{P}_{\mid X_0}$, and $i_{V,X_0}$ is a $(V,X_0)$-vetoer, so $b$ does not defeat $a$ in $\mathbf{P}_{\mid X_0}$, which by VIIA implies that $b$ does not defeat $a$ in $\mathbf{P}$. Thus, $i_{V,X_0}$ is a $(V,X)$-vetoer on $(a,b)$, which by Neutrality implies that $i_{V,X_0}$ is a $(V,X)$-vetoer.\footnote{For suppose $i_{V,X_0}$ is not a $(V,X)$-vetoer, so there are  $a',b'\in X$ and a $(V,X)$-profile $\mathbf{Q}$ such that $i_{V,X_0}$ ranks $a'$ above $b'$ in $\mathbf{Q}$ but $b'$ defeats $a'$ in $\mathbf{Q}$. Consider any permutation $\sigma$ of $X$ such that $\sigma(a)=a'$ and $\sigma(b)=b'$. Applying this permutation to $\mathbf{Q}$ as in Footnote \ref{Permutation} yields a profile $\sigma\mathbf{Q}$ in which $i_{V,X_0}$ ranks $a$ above $b$.  By the permutation version of Neutrality in Footnote \ref{Permutation}, since $b'$ defeats $a'$ in $\mathbf{Q}$, $b$ defeats $a$ in $\sigma\mathbf{Q}$. This contradicts the fact that $i_{V,X_0}$ is a $(V,X)$-vetoer on $(a,b)$.}   Now let $\mathbf{Q}^+$ be any $(V,X)$-profile extending $\mathbf{Q}$. Then $x\mathbf{Q}_{i_{V,X_0}}y$ implies $x\mathbf{Q}^+_{i_{V,X_0}}y$, so $y$ does not defeat $x$ in $\mathbf{Q}^+$ according to $f$ because $i_{V,X_0}$ is a $(V,X)$-vetoer. Hence by VIIA, $y$ does not defeat $x$ in $\mathbf{Q}$ according to $f$, which completes the proof.
\end{proof}

\noindent Theorem \ref{BlauDebThm2} shows how moving to the variable-candidate setting and interpreting IIA as VIIA can strengthen impossibility theorems. But by weakening VIIA to Coherent IIA, impossibility results like Theorem \ref{BlauDebThm2} disappear. Split Cycle satisfies Coherent IIA, Availability, Neutrality, and Monotonicity but has no vetoer. 

\subsection{Impossibility, independence, and choice consistency}\label{AlphaVoting}

Our rejection of VIIA in favor of Coherent IIA also leads us to reject another well-known principle that is related to VIIA, at least under one interpretation. In particular, the term `Independence of Irrelevant Alternatives' is sometimes used in the theory of rational choice for a condition that differs from Arrow's but also leads to impossibility theorems when applied in a certain way to voting, as explained below. 

A \textit{choice function} on a set $X$ is a function $\mathcal{C}$ such that for any nonempty subset $Y$ of $X$, $\mathcal{C}(Y)$ is a nonempty subset of $Y$. The intuitive interpretation is that in a given choice situation, it may be that not all alternatives in $X$ are available or feasible; if $Y$ is the set of feasible alternatives, called a \textit{menu} or \textit{feasible set}, the decision maker decides that the ultimately chosen alternative comes from $\mathcal{C}(Y)$. If $\mathcal{C}(Y)$ is a singleton set,  the decision maker has chosen the alternative in that set, whereas if $\mathcal{C}(Y)$ is not a singleton, then some further choice process or tiebreaking mechanism is applied (see \citealt[pp.~14-15]{Schwartz1986}).

A choice function satisfies Sen's \citeyearpar{Sen1971} condition $\alpha$ if
\[\mbox{for all nonempty $Z\subseteq Y\subseteq X$, we have $Z\cap \mathcal{C}(Y)\subseteq \mathcal{C}(Z)$}.\]
As a famous illustration of this condition in the context of individual choice, attributed to Sidney Morgenbesser, imagine that when offered a choice between apple pie and blueberry pie, you choose apple pie; but when offered a choice between apple, blueberry, and cherry, you switch to blueberry. This violates $\alpha$ where $Z=\{\mbox{apple, blueberry}\}$, $Y=\{\mbox{apple, blueberry, cherry}\}$, $\mathcal{C}(Y)=\{\mbox{blueberry}\}$, and $\mathcal{C}(Z)=\{\mbox{apple}\}$. Sen's $\alpha$ is also known as `Chernoff's axiom' (\citealt{Chernoff1954}) and sometimes `Independence of Irrelevant Alternatives' (cf.~\citealt{Radner1954}). But as Suzumura warns \citeyearpar[p.~66]{Suzumura1983}, ``[C]are should be taken concerning the occasional unfortunate confusions in the literature between condition I [IIA] and Chernoff's axiom of choice consistency, despite rather obvious contextual differences between them.'' In Appendix \ref{ArrowSection}, we discuss the common allegation that Arrow himself was guilty of this confusion, as it relates to VIIA.

How can a choice-consistency axiom such as Sen's $\alpha$ be applied to voting to compare it with IIA? The first question is: what does a \textit{feasible set} of candidates represent? One natural interpretation is that $X(\mathbf{P})$ is the set of candidates whose names appear on the ballots in the election represented by $\mathbf{P}$, but after the ballots are collected, a candidate might withdraw, become incapacitated, etc., leaving us with a new feasible set $Y\subset X(\mathbf{P})$ of candidates\footnote{Another interpretation (cf.~\citealt{Bordes1991}) is that $X(\mathbf{P})$ is the set of all possible candidates for office, containing possibly millions of individuals, while the feasible set $Y$ is the set of candidates whose names appear on the voters' ballots in a given election. The problem with this interpretation is that in practice voting methods can only input voters' ballots, not their internal preferences for candidates whose names do not appear on ballots.} and the question of how to choose from $Y$ using the voters' rankings of $X(\mathbf{P})$. To answer this question, given an acyclic VCCR $f$ and profile $\mathbf{P}$, there are two ways to use $f$ and $\mathbf{P}$ to define a choice function on $X(\mathbf{P})$:\footnote{For $\mathcal{G}_f(\mathbf{P},Y)$ to be a choice function, i.e., for $ \varnothing \neq Y\subseteq X(\mathbf{P})$ to imply $\mathcal{G}_f(\mathbf{P},Y)\neq\varnothing$, $f$ must be acyclic. But for $\mathcal{L}_f(\mathbf{P},\cdot)$  to be a choice function, it suffices that $f$ satisfies the weaker axiom of Availability.}
\begin{enumerate}
\item the global choice function $\mathcal{G}_f(\mathbf{P},\cdot)$ induced by $f,\mathbf{P}$: for any nonempty $Y\subseteq X(\mathbf{P})$, \[\mathcal{G}_f(\mathbf{P},Y)=\{y\in Y\mid \mbox{there is no $z\in Y$ that defeats $y$ in $\mathbf{P}$ according to $f$}\}.\]
\item the local choice function $\mathcal{L}_f(\mathbf{P},\cdot)$ induced by $f,\mathbf{P}$: for any nonempty $Y\subseteq X(\mathbf{P})$,
\[\mathcal{L}_f(\mathbf{P},Y)=\{y\in Y\mid \mbox{there is no $z\in Y$ that defeats $y$ in $\mathbf{P}_{\mid Y}$ according to $f$}\}.\]
\end{enumerate}
Intuitively, the local choice function chooses from the feasible set $Y$ by first erasing the names of candidates who have withdrawn, becoming incapacitated, etc., from each voter's ballot and then applying the voting method to the restricted profile $\mathbf{P}_{\mid Y}$. By contrast, while the global choice function excludes the unfeasible candidates from contention, it does not erase their names from voter's ballots, so voters' rankings of balloted but ultimately unfeasible candidates may still affect which of the other candidates are chosen from $Y$.

\begin{example}\label{BordaExample} The distinction between the global choice function and local choice function can be illustrated by the well-known distinction between global Borda count and local Borda count.\footnote{This terminology is due to Kelly \citeyearpar[p.~71, 74]{Kelly1988}. Sen  \citeyearpar[pp.~78-9]{Sen1977} uses the terms `broad' and `narrow'.} Let $f$ be the Borda VCCR according to which $x$ defeats $y$ in a profile $\mathbf{P}$ according to $f$ just in case the Borda score of $x$ in $\mathbf{P}$ is greater than that of $y$. Then $\mathcal{G}_f(\mathbf{P},Y)$, the elements of $Y$ chosen according to global Borda count, are the elements of $Y$ whose Borda scores are maximal among elements of $Y$, where Borda scores are calculated with respect to the full profile $\mathbf{P}$. By contrast,    $\mathcal{L}_f(\mathbf{P},Y)$, the element of $Y$ chosen according to local Borda count, are the elements of $Y$ whose Borda scores are maximal among elements of $Y$, where Borda scores are calculated with respect to the restricted profile $\mathbf{P}_{\mid Y}$. For a concrete example, consider the following profiles $\mathbf{P}$ and $\mathbf{P}_{\mid \{x,y,a\}}$:
\begin{center}
$\mathbf{P}$\qquad
\begin{tabular}{ccc}
 $1$ & $1$ & $2$    \\\hline
$x$ & $y$ & $ y$    \\
$a$ &  $x$ & $x$   \\
$b$ &  $a$ & $c$  \\
$c$ &  $b$ & $b$  \\
$y$ &  $c$ & $a$ \\
\end{tabular}\qquad\qquad\qquad$\mathbf{P}_{\mid \{x,y,a\}}$\qquad\begin{tabular}{ccc}
 $1$ & $1$ & $2$    \\\hline
$x$ & $y$ & $ y$    \\
$a$ &  $x$ & $x$   \\
$y$ &  $a$ & $a$ 
\end{tabular}\end{center}
Global Borda count yields $\mathcal{G}_f(\mathbf{P},\{x,y,a\})=\{x\}$, as $x$ has the highest Borda score in $\mathbf{P}$, while local Borda count yields $\mathcal{L}_f(\mathbf{P},\{x,y,a\})=\{y\}$, as $y$ has the highest Borda score in the restricted profile $\mathbf{P}_{\mid \{x,y,a\}}$.\end{example}

Before using the global and local choice functions to define two senses of Sen's $\alpha$ for voting, we must note that the distinction between global and local is lost under the assumption of VIIA.

\begin{proposition}\label{GlobalEqualsLocal} Let $f$ be an acyclic VCCR. The following are equivalent:
\begin{enumerate}
\item\label{GlobalEqualsLocal1} $f$ satisfies VIIA;
\item\label{GlobalEqualsLocal2} for any profile $\mathbf{P}$ and $Y\subseteq X(\mathbf{P})$, $\mathcal{G}_f(\mathbf{P},Y)=\mathcal{L}_f(\mathbf{P},Y)$.
\end{enumerate}
\end{proposition}
\begin{proof} Suppose $f$ satisfies VIIA. Consider any profile $\mathbf{P}$ and $Y\subseteq X(\mathbf{P})$. Then the following are equivalent for any $y,z\in Y$:
\begin{itemize}
\item $z$ defeats $y$ in $\mathbf{P}$ according to $f$;
\item $z$ defeats $y$ in $\mathbf{P}_{\mid \{y,z\}}$ according to $f$ (by VIIA);
\item $z$ defeats $y$ in $(\mathbf{P}_{\mid Y})_{\mid \{y,z\}}$ according to $f$ (since $\mathbf{P}_{\mid \{y,z\}}=(\mathbf{P}_{\mid Y})_{\mid \{y,z\}}$);
\item $z$ defeats $y$ in $\mathbf{P}_{\mid Y}$ according to $f$ (by VIIA).
\end{itemize}
Hence $\mathcal{G}_f(\mathbf{P},Y)=\mathcal{L}_f(\mathbf{P},Y)$ by the definitions of $\mathcal{G}_f(\mathbf{P},\cdot)$ and $\mathcal{L}_f(\mathbf{P},\cdot)$.

Suppose condition \ref{GlobalEqualsLocal2} holds. Consider profiles $\mathbf{P}$ and $\mathbf{P}'$ with $x,y\in X(\mathbf{P})\cap X(\mathbf{P}')$ and $\mathbf{P}_{\mid\{x,y\}}=\mathbf{P}'_{\mid\{x,y\}}$. We must show that $x$ defeats $y$ in $\mathbf{P}$ according to $f$ if and only if $x$ defeats $y$ in $\mathbf{P}'$ according to $f$. This is equivalent to the claim that $\mathcal{G}_f(\mathbf{P},\{x,y\})=\mathcal{G}_f(\mathbf{P}',\{x,y\})$. We claim that the following equations hold:
\begin{center}\begin{tabular}{ccc}
$\mathcal{G}_f(\mathbf{P},\{x,y\})$ & & $\mathcal{G}_f(\mathbf{P}',\{x,y\})$\\
$\rotatebox[origin=c]{-90}{$=$}$ && $\rotatebox[origin=c]{-90}{$=$}$ \\
$\mathcal{L}_f(\mathbf{P},\{x,y\})$ & = & $\mathcal{L}_f(\mathbf{P}',\{x,y\})$
\end{tabular}\end{center}
The vertical equations hold by condition \ref{GlobalEqualsLocal2}, while the horizontal equation holds since $\mathbf{P}_{\mid\{x,y\}}=\mathbf{P}'_{\mid\{x,y\}}$. Hence $\mathcal{G}_f(\mathbf{P},\{x,y\})=\mathcal{G}_f(\mathbf{P}',\{x,y\})$.\end{proof}

Assuming we weaken VIIA, we can make the local-global distinction and hence distinguish two senses of Sen's $\alpha$ in the context of voting.

\begin{definition} Let $f$ be an acyclic VCCR.
\begin{enumerate}
\item $f$ satisfies Global-$\alpha$  if $\mathcal{G}_f(\mathbf{P},\cdot)$ satisfies $\alpha$ for all profiles $\mathbf{P}$;
\item $f$ satisfies Local-$\alpha$ if $\mathcal{L}_f(\mathbf{P},\cdot)$ satisfies $\alpha$ for all profiles $\mathbf{P}$.
\end{enumerate}
\end{definition}

To illustrate this definition, let us return to the discussion of Borda count from Example \ref{BordaExample}.

\begin{example} To see that the Borda VCCR $f$ satisfies Global-$\alpha$, we must check that 
\[\mbox{for all profiles $\mathbf{P}$ and nonempty $Z\subseteq Y\subseteq X(\mathbf{P})$, we have $Z\cap \mathcal{G}_f(\mathbf{P},Y)\subseteq \mathcal{G}_f(\mathbf{P},Z)$}.\]
Indeed, if $x\in Z$ and $x\in \mathcal{G}_f(\mathbf{P},Y)$, so among the candidates in $Y$, $x$ has maximal Borda score calculated with respect to the full profile $\mathbf{P}$, then since $Z\subseteq Y$, it follows that among the candidates in $Z$, $x$ has maximal Borda score calculated with respect to the full profile $\mathbf{P}$, so $x\in \mathcal{G}_f(\mathbf{P},Z)$. By contrast, the profile $\mathbf{P}$ from Example \ref{BordaExample} shows that the Borda VCCR does not satisfy Local-$\alpha$: 
\[\mbox{$\{x,y,a\}\subseteq \{x,y,a,b,c\}\subseteq X(\mathbf{P})$, but $\{x,y,a\}\cap \mathcal{L}_f(\mathbf{P},\{x,y,a,b,c\})\not\subseteq \mathcal{L}_f(\mathbf{P},\{x,y,a\})$},\]
because $x\in \{x,y,a\}\cap \mathcal{L}_f(\mathbf{P},\{x,y,a,b,c\})$ but $x\not\in \mathcal{L}_f(\mathbf{P},\{x,y,a\})$. While $x$ has the highest Borda score calculated with respect to the full profile $\mathbf{P}$, so $\mathcal{L}_f(\mathbf{P},\{x,y,a,b,c\})=\{x\}$,  $y$ has the highest Borda score calculated with respect to the restricted profile $\mathbf{P}_{\mid \{x,y,a\}}$, so $\mathcal{L}_f(\mathbf{P},\{x,y,a\})=\{y\}$.\end{example}

It is no accident that the Borda VCCR satisfies Global-$\alpha$ but not Local-$\alpha$. It is well known that Global-$\alpha$ imposes no constraint on an acyclic VCCR.

\begin{proposition} If $f$ is an acyclic VCCR, then $f$ satisfies Global-$\alpha$.
\end{proposition}

\begin{proof} The claim that $f$ satisfies Global-$\alpha$ is the claim that for any profile $\mathbf{P}$ and nonempty $Z\subseteq Y\subseteq X(\mathbf{P})$, we have $Z\cap \mathcal{G}_f(\mathbf{P},Y)\subseteq \mathcal{G}_f(\mathbf{P},Z)$. Indeed, if $y\in Z\cap \mathcal{G}_f(\mathbf{P},Y)$, so by definition there is no $z\in Y$ that defeats $y$ in $\mathbf{P}$ according to $f$, then since $Z\subseteq Y$, there is no $z\in Z$ that defeats $y$ in $\mathbf{P}$ according to $f$, which by definition implies $y\in \mathcal{G}_f(\mathbf{P},Z)$.\end{proof}

Let us now consider Local-$\alpha$ as a constraint on VCCRs. First, we note that it is a weakening of VIIA.

\begin{proposition}\label{VIIAalpha}$\,$
\begin{enumerate}
\item\label{VIIAalpha1} If $f$ is an acyclic VCCR satisfying VIIA, then $f$ satisfies Local-$\alpha$;
\item\label{VIIAalpha2} There are acyclic VCCRs satisfying Local-$\alpha$ but not FIIA and hence not VIIA.
\end{enumerate}
\end{proposition}
\begin{proof} For part \ref{VIIAalpha1}, assuming that $f$ satisfies VIIA, we show that for any profile $\mathbf{P}$ and nonempty $Z\subseteq Y\subseteq X(\mathbf{P})$, we have $Z\cap \mathcal{L}_f(\mathbf{P},Y)\subseteq \mathcal{L}_f(\mathbf{P},Z)$.  Suppose $y\in Z$ but $y\not\in \mathcal{L}_f(\mathbf{P},Z)$, so there is an $x\in Z$ that defeats $y$ in $\mathbf{P}_{\mid Z}$ according to $f$. Then since  $(\mathbf{P}_{\mid Z})_{\mid \{x,y\}}=(\mathbf{P}_{\mid Y})_{\mid \{x,y\}}$, it follows by VIIA that $x$ defeats $y$ in $\mathbf{P}_{\mid Y}$ according to $f$, which with $x\in Z\subseteq Y$ implies $y\not\in \mathcal{L}_f(\mathbf{P},Y)$.

For part \ref{VIIAalpha2}, let $f$ be a VCCR such that $x$ defeats $y$ in $\mathbf{P}$ according to $f$ if and only if (i) $x$ is unanimously preferred to $y$ and (ii) there is a $z\in X(\mathbf{P})\setminus\{x,y\}$ such that $x\to z$ but $y\not\to z$. Then $f$ is acyclic in virtue of (i). To see that $f$ violates FIIA, consider two profiles $\mathbf{P}$ and $\mathbf{P}'$ with $X(\mathbf{P})=X(\mathbf{P}')=\{x,y,z\}$ and $V(\mathbf{P})=V(\mathbf{P}')$ such that all $i\in V(\mathbf{P})$ have  $x\mathbf{P}_i z \mathbf{P}_iy$ while all $i\in V(\mathbf{P}')$ have $x\mathbf{P}'_i y \mathbf{P}'_iz$. Then $\mathbf{P}_{\mid \{x,y\}}=\mathbf{P}'_{\mid\{x,y\}}$, and $x$ defeats $y$ in $\mathbf{P}$ but not in $\mathbf{P}'$, violating FIIA. To see that $f$ satisfies Local-$\alpha$, we must show that for any profile $\mathbf{P}$ and nonempty $Z\subseteq Y\subseteq X(\mathbf{P})$, we have $Z\cap \mathcal{L}_f(\mathbf{P},Y)\subseteq \mathcal{L}_f(\mathbf{P},Z)$. Suppose $y\in Z$ but $y\not\in \mathcal{L}_f(\mathbf{P},Z)$, so there is an $x\in Z$ that defeats $y$ in $\mathbf{P}_{\mid Z}$ according to $f$. Hence $x$ is unanimously preferred to $y$ in $\mathbf{P}_{\mid Z}$ and there is a $z\in Z$ such that $x\to z$ but $y\not\to z$. Then since $Z\subseteq Y$, we have that $x\in Y$, that $x$ is unanimously preferred to $y$ in $\mathbf{P}_{\mid Y}$ according to $f$, and that there is a $z\in Y$ such that $x\to z$ but $y\not\to z$. Therefore, $x$ defeats $y$ in $\mathbf{P}_{\mid Y}$ according to $f$, which with $x\in Y$ implies $y\not\in \mathcal{L}_f(\mathbf{P},Y)$.\end{proof}

Although weaker than VIIA, Local-$\alpha$ is still a significant restriction on an acyclic VCCR, as it rules out that the VCCR coincides with majority rule on two-candidates profiles.

\begin{definition} A VCCR $f$ satisfies \textit{Binary Majoritarianism} if for any profile $\mathbf{P}$ with $X(\mathbf{P})=\{x,y\}$, $x$ defeats $y$ in $\mathbf{P}$ according to $f$ if and only if $x$ is majority preferred to $y$ in $\mathbf{P}$.\end{definition}

We have the following easy impossibility result.

\begin{proposition}\label{BinaryImpossibility} There is no VCCR satisfying Local-$\alpha$, Availability, and Binary Majoritarianism.
\end{proposition}

\begin{proof} Consider a profile $\mathbf{P}$ with $X(\mathbf{P})=\{x,y,z\}$ and a majority cycle $x\to y\to z\to x$. By Availability, there is some $a\in \{x,y,z\}$ who is undefeated in $\mathbf{P}$ according to $f$. Since there is a majority cycle, there is some $b\in \{x,y,z\}$ such that $b\to a$. Hence by Binary Majoritarianism, $b$ defeats $a$ in $\mathbf{P}_{\mid \{a,b\}}$ according to $f$. Thus, we have $a\in \{a,b\}\cap \mathcal{L}_f(\mathbf{P}, \{x,y,z\})$ but $a\not\in \mathcal{L}_f(\mathbf{P}_{\mid \{a,b\}},\{a,b\})$, so $f$ violates Local-$\alpha$.\end{proof}

Combining Propositions \ref{VIIAalpha}.\ref{VIIAalpha1} and \ref{BinaryImpossibility}, we have the analogue of Proposition \ref{BinaryImpossibility} under VIIA.

\begin{corollary} There is no VCCR satisfying VIIA, Availability, and Binary Majoritarianism.
\end{corollary}

Finally, let us come full circle and return to voting methods in the sense of Definition \ref{VotingMethod}, as opposed to VCCRs. What are the implications of the impossibility result above for voting methods? To answer this question, we first adapt the definition of $\alpha$ to voting methods.

\begin{definition} A voting method $F$ satisfies $\overline{\alpha}$ if for all nonempty $Z\subseteq X(\mathbf{P})$,  $Z\cap F(\mathbf{P})\subseteq F(\mathbf{P}_{\mid Z})$.\footnote{\label{alphabarnote}Note that this is equivalent to the more direct translation of $\alpha$ to voting methods: for all nonempty $Z\subseteq Y\subseteq X(\mathbf{P})$,  $Z\cap F(\mathbf{P}_{\mid Y})\subseteq F(\mathbf{P}_{\mid Z})$.}
\end{definition}
\noindent Simply put, $\overline{\alpha}$ states that if $x$ is a winner in an election represented by $\mathbf{P}$, then $x$ would also have been a winner had some other candidates not participated in the election at all, while the remaining candidates---those in $Z$---were ranked in the same way, as represented by $\mathbf{P}_{\mid Z}$. Although a candidate's never appearing on the ballot in the first place is conceptually different than their appearing on the ballot but then withdrawing or being incapacitated after the ballots are collected---as in our interpretation of the choice functional setting above---nonetheless, there is a formal connection: the satisfaction of $\overline{\alpha}$ by a voting method is equivalent to the satisfaction of Local-$\alpha$ by any VCCR that defeat rationalizes the voting method.

\begin{lemma}\label{AlphaAlpha} If $F$ is a voting method that is defeat rationalized by a VCCR $f$, then $F$ satisfies $\overline{\alpha}$ if and only if $f$ satisfies Local-$\alpha$.
\end{lemma}

\begin{proof} Since $F$ is defeat rationalized by $f$, $F(\mathbf{P})$ is the set of undefeated candidates in $\mathbf{P}$ according to $f$, which is also equal to $\mathcal{L}_f(\mathbf{P},X(\mathbf{P}))$. Thus, for all nonempty $Z\subseteq Y\subseteq X(\mathbf{P})$, we have:
\begin{eqnarray*}
&&Z\cap F(\mathbf{P}_{\mid Y})\subseteq F(\mathbf{P}_{\mid Z}) \\
&\Leftrightarrow & Z\cap \mathcal{L}_f(\mathbf{P},Y)\subseteq \mathcal{L}_f(\mathbf{P},Z).
\end{eqnarray*} 
Hence $F$ satisfies $\overline{\alpha}$ (in the equivalent form given in Footnote \ref{alphabarnote}) if and only if $f$ satisfies Local-$\alpha$.
\end{proof}

We can now answer our question about voting methods with the following impossibility result, whose proof is easily obtained by adapting that of Proposition \ref{BinaryImpossibility}.

\begin{proposition} There is no voting method satisfying $\overline{\alpha}$ and Binary Majoritarianism. 
\end{proposition}

As a voting method (resp.~VCCR) Split Cycle satisfies Binary Majoritarianism but not $\overline{\alpha}$ (resp.~Local-$\alpha$). The mistake of insisting on  $\overline{\alpha}$ for voting is essentially the same as the mistake of insisting on IIA, which can be seen by reformulating $\overline{\alpha}$ as follows: for all nonempty $Z\subseteq X(\mathbf{P})$,  if $x\in Z$ but $x\not\in F(\mathbf{P}_{\mid Z})$, then $x\not\in F(\mathbf{P})$. In terms of a defeat rationalization of $F$, this means that if $x$ is defeated in the smaller profile $\mathbf{P}_{\mid Z}$ then $x$ must also be defeated in the larger profile $\mathbf{P}$. But this should not follow if the larger profile $\mathbf{P}$ is more incoherent than $\mathbf{P}_{\mid Z}$. If $\mathbf{P}$ is sufficiently incoherent, we may need to suspend judgment on many defeat relations that we could coherently accept in $\mathbf{P}_{\mid Z}$.

\begin{example} The same example used against IIA in Example \ref{IIAExample} can be adapted to argue against Local-$\alpha$ for VCCRs or $\overline{\alpha}$ for voting methods. Consider the following profiles $\mathbf{P}$ and $\mathbf{P}'$:

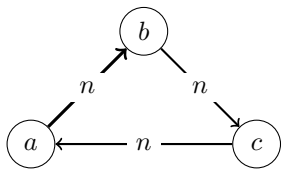
\begin{center}
$\mathbf{P}$\qquad
\begin{minipage}{2in}\begin{tabular}{ccc}
$n$ & $n$ & $n$   \\\hline
$a$ & $b$ &  $a$ \\
$b$ &  $a$ & $b$ \\
\end{tabular}\end{minipage}\begin{minipage}{2in}\begin{tikzpicture}

\node[circle,draw, minimum width=0.25in] at (0,0) (a) {$a$}; 
\node[circle,draw,minimum width=0.25in] at (3,0) (b) {$b$}; 

\path[->,draw,thick] (a) to node[fill=white] {$n$} (b);

\end{tikzpicture}
\end{minipage}\vspace{.15in}

$\mathbf{P}'$\qquad
\begin{minipage}{2in}\begin{tabular}{ccc}
$n$ & $n$ & $n$   \\\hline
$\boldsymbol{a}$ & $\boldsymbol{b}$ &  $c$ \\
$\boldsymbol{b}$ &  $c$ & $\boldsymbol{a}$ \\
$c$ &  $\boldsymbol{a}$ &  $\boldsymbol{b}$ \\
\end{tabular}\end{minipage}\begin{minipage}{2in}\begin{tikzpicture}

\node[circle,draw, minimum width=0.25in] at (0,0) (a) {$a$}; 
\node[circle,draw,minimum width=0.25in] at (3,0) (c) {$c$}; 
\node[circle,draw,minimum width=0.25in] at (1.5,1.5) (b) {$b$}; 

\path[->,draw,very thick] (a) to node[fill=white] {$n$} (b);
\path[->,draw,thick] (b) to node[fill=white] {$n$} (c);
\path[->,draw,thick] (c) to node[fill=white] {$n$} (a);

\end{tikzpicture}
\end{minipage}\end{center}
In the context of the perfectly coherent profile $\mathbf{P}$, the margin of $n$ for $a$ over $b$ should be sufficient for $a$ to defeat $b$, so $a$ should be the uniquely chosen winner. But in the context of the incoherent profile $\mathbf{P}'$, it is not sufficient: no one can be judged to defeat anyone else---this follows from Anonymity, Neutrality, and Availability---so all three must be included in the choice set to which a further tiebreaking process is applied. This is a counterexample to Local-$\alpha$ and $\overline{\alpha}$: $b$ is undefeated in $\mathbf{P}'$ but not in $\mathbf{P}'_{\mid \{a,b\}}=\mathbf{P}$.\end{example}

Our conclusion concerning $\alpha$ applied to voting is in the spirit of Sen's \citeyearpar{Sen1993} view that ``Violations of property $\alpha$\dots can be related to various different types of reasons---easily understandable when the external context is spelled out'' (p.~501). Sen (pp.~500-502) focuses on rationalizing violations of $\alpha$ in individual choice by reference to features of the context of choice. Here we have focused on rationalizing violations in voting by reference to features of the context given by the profile---namely, an increase in incoherence from one profile to another. To overlook this context would be to commit The Fallacy of IIA from Section~\ref{CoherentIIASection}.

\section{Conclusion}\label{Conclusion}

The pessimistic conclusions about democracy that some have drawn from the Paradox of Voting and Arrow's Impossibility Theorem are not justified. Like most voting theorists, we are more optimistic. In particular, we believe that many majority cycles can be resolved in a rational way, while respecting the principle of Majority Defeat, as shown by Split Cycle. Of course there remain some cycles, such as a perfect cycle $a\to b\to c\to a$ in which each candidate is majority preferred to the next by exactly the same margin (and there are no other candidates), which must lead to a tie between all candidates. But to think that democracy is devastated by the possibility of such ties seems almost as unreasonable as thinking that democracy is devastated by the possibility that in an election with only two candidates and an even number of voters, it could happen that half of the voters vote for $x$ over $y$ while half vote for $y$ over $x$. The fact that the set of winners cannot always be a singleton---that some further tiebreaking mechanism must be in place---hardly warrants very pessimistic conclusions, especially if the probability of having many tied candidates is sufficiently low, as we expect when there are many voters compared to candidates (see \citealt{HP2020}). 

Far from justifying pessimism about democracy, social choice theory leads the way to voting procedures that can improve democratic decision making. We agree with Maskin and Sen \citeyearpar{Maskin2017a,Maskin2017b} that a major improvement would come in replacing Plurality voting with a voting procedure using ranked ballots that elects a Condorcet winner whenever there is one. In this paper, we have arrived at a unique collective choice rule, Split Cycle, via six axioms concerning when one candidate should defeat another in a democratic election---with the key axiom being the axiom of Coherent IIA that weakens Arrow's IIA and explains why the latter is too strong. As theorists, we sleep well at night knowing that we have a solid theoretical justification for handling majority cycles in a certain way should they arise. As citizens and committee members, we hope that in practice our elections will have Condorcet winners and that we will elect them.

\subsection*{Acknowledgements}

We thank Mikayla Kelley, John Patty, and the two anonymous referees for the \textit{Journal of Theoretical Politics} for helpful comments. We are also grateful for useful feedback received at the Work in Progress Seminar and Logic Seminar at the University of Maryland in July 2020 and at the FERC reading group at UC Berkeley in August 2020.

\appendix 

\section{Proofs for Section \ref{SCsection}}\label{Proofs}

\Reformulation*

\begin{proof} Suppose that in $\mathbf{P}$, $x$ wins by more than $n$ over $y$ for the smallest natural number $n$ such that there is no majority cycle, containing $x$ and $y$, in which each candidate wins by more than $n$ over the next candidate in the cycle. Then $Margin_\mathbf{P}(x,y)>n\geq 0$. Now consider some majority cycle $\rho$ in $\mathbf{P}$ containing $x$ and $y$. By our choice of $n$, we have $n\geq Split\#_\mathbf{P}(\rho)$, so $Margin_\mathbf{P}(x,y)>n$ implies $Margin_\mathbf{P}(x,y)> Split\#_\mathbf{P}(\rho)$. 

Conversely, suppose $Margin_\mathbf{P}(x,y)>0$ and $Margin_\mathbf{P}(x,y)>Split\#_\mathbf{P}(\rho)$ for every majority cycle $\rho$ in $\mathbf{P}$ containing $x$ and $y$. If there exist such cycles, let $n$ be the maximum of their splitting numbers, and otherwise let $n=0$. It follows that there is no majority cycle containing $x$ and $y$ in which each candidate wins by \textit{more than} $n$ over the next candidate in the cycle; moreover, $n$ is the smallest natural number for which this holds. By our initial supposition, $Margin_\mathbf{P}(x,y)>n$, so we are done.\end{proof}

\OnlySome*

\begin{proof} We use the formulation of Split Cycle in Lemma \ref{SplittingLem}. If $Margin_\mathbf{P}(x,y)>Split\#_\mathbf{P}(\rho)$ for every majority cycle $\rho$ in $\mathbf{P}$ containing $x$ and $y$, then in particular  $Margin_\mathbf{P}(x,y)>Split\#_\mathbf{P}(\rho)$ for every majority cycle $\rho$ in $\mathbf{P}$ of the form $x \rightarrow y\rightarrow z_1\rightarrow \dots\rightarrow z_n\rightarrow x$. Conversely, suppose $Margin_\mathbf{P}(x,y)>Split\#_\mathbf{P}(\rho)$ for every majority cycle $\rho$ in $\mathbf{P}$ of the form $x \rightarrow y\rightarrow z_1\rightarrow \dots\rightarrow z_n\rightarrow x$. Now consider a majority cycle $\rho$ in $\mathbf{P}$ containing $x$ and $y$, whose splitting number is maximal among all such majority cycles. We must show $Margin_\mathbf{P}(x,y)>Split\#_\mathbf{P}(\rho)$. If $y$ occurs immediately after $x$ in $\rho$, then by ``rotating the cycle'' we obtain a cycle $\rho'$ of the form $x \rightarrow y\rightarrow z_1\rightarrow \dots\rightarrow z_n\rightarrow x$ with the same splitting number as $\rho$, in which case $Margin_\mathbf{P}(x,y)>Split\#_\mathbf{P}(\rho')$ by our initial supposition and hence $Margin_\mathbf{P}(x,y)>Split\#_\mathbf{P}(\rho)$. Thus, suppose $y$ does not occur immediately after $x$ in $\rho$. Then without loss of generality, we may assume $\rho$ is of the form  $y\to z_1\to\dots \to z_n\to x\to w_1\to\dots\to w_m\to y$. Let $\rho'$ be the cycle $x\to y\to z_1\to\dots \to z_n\to x$. Then $Margin_\mathbf{P}(x,y)>Split\#_\mathbf{P}(\rho')$ by our initial supposition, so the splitting number of $\rho'$ is the margin associated with some successive candidates in the sequence $y,z_1,\dots z_n,x$. Since $y,z_1,\dots z_n,x$ is a subsequence of $\rho$, and the splitting number is defined as a minimum, it follows that $Split\#_\mathbf{P}(\rho')\geq Split\#_\mathbf{P}(\rho)$. Then since $Margin_\mathbf{P}(x,y)>Split\#_\mathbf{P}(\rho')$, we have $Margin_\mathbf{P}(x,y)>Split\#_\mathbf{P}(\rho)$, and since $\rho$ was chosen to have maximal splitting number among all majority cycles containing $x$ and $y$, we are~done.\end{proof}

\section{Arrow's alleged confusion and VIIA}\label{ArrowSection}

Arrow has been accused of confusing his own condition of IIA, an interprofile condition, with a choice-consistency condition such as Sen's $\alpha$, defined in Section \ref{AlphaVoting} (see, e.g.,  \citealt[\S~3]{Hansson1973}, \citealt[p.~989]{Ray1973}, and \citealt[p.~97, endnote~16, p.~250]{Suzumura1983}). To clarify this matter, which is relevant to our distinction between FIIA and VIIA, we first note that Arrow did not state IIA in what is now its most common form, given in Definition \ref{IIAdef}.\ref{IIAdef1}. Instead, he stated it in the equivalent form (assuming acylicity) in Definition~\ref{ChoiceIIAdef}.\ref{ChoiceIIAdef1}.

\begin{definition}\label{ChoiceIIAdef} Let $f$ be an acyclic VCCR. 
\begin{enumerate}
\item\label{ChoiceIIAdef1}  $f$ satisfies \textit{global choice FIIA} if for any profiles $\mathbf{P}$ and $\mathbf{P}'$ with $V(\mathbf{P})=V(\mathbf{P}')$ and $X(\mathbf{P})=X(\mathbf{P}')$ and $Y\subseteq X(\mathbf{P})$, if $\mathbf{P}_{\mid Y}=\mathbf{P}'_{\mid Y}$, then  $\mathcal{G}_f(\mathbf{P},Y)= \mathcal{G}_f(\mathbf{P}',Y)$.
\item $f$ satisfies \textit{local choice FIIA} if for any profiles $\mathbf{P}$ and $\mathbf{P}'$ with $V(\mathbf{P})=V(\mathbf{P}')$ and $X(\mathbf{P})=X(\mathbf{P}')$ and $Y\subseteq X(\mathbf{P})$, if $\mathbf{P}_{\mid Y}=\mathbf{P}'_{\mid Y}$, then  $\mathcal{L}_f(\mathbf{P},Y)= \mathcal{L}_f(\mathbf{P}',Y)$.
\end{enumerate}
\end{definition}

\begin{proposition} Let $f$ be an acyclic VCCR. Then $f$ satisfies global choice FIIA if and only if $f$ satisfies FIIA.
\end{proposition}
\begin{proof} Assume $f$ satisfies global choice FIIA. To show that $f$ satisfies FIIA, suppose $\mathbf{P}_{\mid \{x,y\}}=\mathbf{P}'_{\mid \{x,y\}}$. Then by global choice FIIA, $\mathcal{G}_f(\mathbf{P},\{x,y\})= \mathcal{G}_f(\mathbf{P}',\{x,y\})$. It follows by definition of $\mathcal{G}_f$ that $x$ defeats $y$ in $\mathbf{P}$ according to $f$ if and only if $x$ defeats $y$ in $\mathbf{P}'$ according to $f$. Hence $f$ satisfies FIIA.

Now assume $f$ satisfies FIIA. To show that $f$ satisfies global choice FIIA, suppose $\mathbf{P}_{\mid Y}=\mathbf{P}'_{\mid Y}$. Then for any two $x,y\in Y$,  $\mathbf{P}_{\mid \{x,y\}}=\mathbf{P}'_{\mid \{x,y\}}$. Hence by FIIA, $x$ defeats $y$ in $\mathbf{P}$ according to $f$ if and only if $x$ defeats $y$ in $\mathbf{P}'$ according to $f$. It follows by  definition of $\mathcal{G}_f$ that  $\mathcal{G}_f(\mathbf{P},Y)= \mathcal{G}_f(\mathbf{P}',Y)$.\end{proof}

In contrast to global choice FIIA, which is a significant restriction on an acyclic VCCR, local choice FIIA is no restriction.

\begin{proposition} If $f$ is an acyclic VCCR, then $f$ satisfies local choice FIIA.
\end{proposition}
\begin{proof} By definition, we have
\[\mathcal{L}_f(\mathbf{P},Y)=\{y\in Y\mid \mbox{there is no $z\in Y$ that defeats $y$ in $\mathbf{P}_{\mid Y}$ according to $f$}\}\]
\[\mathcal{L}_f(\mathbf{P}',Y)=\{y\in Y\mid \mbox{there is no $z\in Y$ that defeats $y$ in $\mathbf{P}'_{\mid Y}$ according to $f$}\},\]
which with $\mathbf{P}_{\mid Y}=\mathbf{P}'_{\mid Y}$ implies $\mathcal{L}_f(\mathbf{P},Y)= \mathcal{L}_f(\mathbf{P}',Y)$.
\end{proof}

Now consider one of Arrow's \citeyearpar[p.~26]{Arrow1963} supposed arguments for IIA:
\begin{quote}Suppose that an election is held, with a certain number of candidates in the field, each individual filing his list of preferences, and then one of the candidates dies. Surely the social choice should be made by taking each individual's preference lists, blotting out completely the dead candidate's name, and considering only the orderings of the remaining names in going through the procedure of determining the winner. That is, the choice to be made among the set $S$ of surviving candidates should be independent of the preferences of individuals for candidates not in $S$. To assume otherwise would be to make the result of the election dependent on the obviously accidental circumstance of whether a candidate died before or after the date of polling.
\end{quote}
We agree with the literature cited above (\citealt{Hansson1973,Ray1973,Suzumura1983}) that this is not an argument that one's VCCR should satisfy IIA. But neither is it an argument that one's VCCR should satisfy Local-$\alpha$. In our view, the argument above is at most an argument for the thesis that if a candidate who appeared on the ballots in $\mathbf{P}$ dies after the ballots are collected, then one should choose among the surviving candidates using the \textit{local choice function} $\mathcal{L}_f(\mathbf{P},\cdot)$. As long as one chooses using the local choice function, one follows all of Arrow's recommendations above, regardless of whether $f$ satisfies IIA or Local-$\alpha$. 

However, Arrow does not officially make the distinction between the global and local choice functions. He only officially defines the global choice function induced by a CCR.\footnote{This follows from Arrow's \citeyearpar{Arrow1963} Definition 4 (p.~23), the first sentence of his Section III.3 (p.~26), and his Definition 3 (p.~15).} But if in the example above, Arrow wants the global choice function to act like the local choice function, this leads to VIIA according to Proposition \ref{GlobalEqualsLocal}. Thus, one can understand the otherwise puzzling example of the dead candidate as possibly related to Arrow's implicit commitment to VIIA. 

Arrow  \citeyearpar[p.~27]{Arrow1963} gives another supposed argument for IIA, based on Borda count:
\begin{quote}
[S]uppose that there are three voters and four candidates, $x$, $y$, $z$, and $w$. Let the weights for the first, second, third, and fourth choices be 4, 3, 2, and 1, respectively. Suppose that individuals 1 and 2 rank candidates in the order $x$, $y$, $z$, and $w$, while individual 3 ranks them in the order $z$, $w$, $x$, and $y$. Under the given electoral system, $x$ is chosen. Then, certainly, if $y$ is deleted from the ranks of the candidates, the system applied to the remaining candidates should yield the same result, especially since, in this case, $y$ is inferior to $x$ according to the tastes of every individual; but, if $y$ is in fact deleted, the indicated electoral system would yield a tie between $x$ and $z$. 
\end{quote}
Let $\mathbf{P}$ be the initial profile described by Arrow with $X(\mathbf{P})=\{x,y,z,w\}$. When Arrow says ``if $y$ is deleted from the ranks of the candidates, the system applied to the resulting candidates should yield the same result'' which of the following did he mean?
\begin{enumerate}
\item since $\mathcal{G}_f(\mathbf{P},X(\mathbf{P}))=\{x\}$, we should have $\mathcal{G}_f(\mathbf{P},\{x,z,w\})=\{x\}$; 
\item since $\mathcal{G}_f(\mathbf{P},X(\mathbf{P}))=\{x\}$, we should have $\mathcal{G}_f(\mathbf{P}_{\mid \{x,z,w\}},\{x,z,w\})=\{x\}$; 
\item since $\mathcal{L}_f(\mathbf{P},X(\mathbf{P}))=\{x\}$, we should have $ \mathcal{L}_f(\mathbf{P},\{x,z,w\})=\{x\}$; 
\item since $ \mathcal{L}_f(\mathbf{P},X(\mathbf{P}))=\{x\}$, we should have $ \mathcal{L}_f(\mathbf{P}_{\mid \{x,z,w\}},\{x,z,w\})=\{x\}$.
\end{enumerate}
In fact, options 2, 3, and 4 are equivalent. Since Arrow only officially discusses the global choice function induced by a CCR, he could not have officially meant 3 or 4. Moreover, since Arrow assumes that all of the profiles in the domain of a given SWF have the same set of candidates,  he could not have officially meant 2, which requires that both $\mathbf{P}$ and $\mathbf{P}_{\mid \{x,z,w\}}$ be in the domain of $f$. Thus, only option 1 officially makes sense in his framework. Yet if $f$ is Borda count, then $\mathcal{G}_f$ is global Borda count, which  still chooses $x$ as the unique winner after $y$ is removed from the input to $\mathcal{G}_f(\mathbf{P},\cdot)$, contradicting Arrow's conclusion. 

Arrow's passage above certainly shows that the Borda VCCR $f$ violates VIIA, because $x$ defeats $z$ in $\mathbf{P}$ according to $f$ but not in  $\mathbf{P}_{\mid \{x,z,w\}}$.\footnote{It also shows that the Borda VCCR $f$ violates what could be called \textit{Local-$\beta$} (see \citealt{Sen1971}): for any profile $\mathbf{P}$ and nonempty $Z\subseteq Y\subseteq X(\mathbf{P})$, if $\mathcal{L}_f(\mathbf{P},Z)\cap \mathcal{L}_f(\mathbf{P},Y)\neq\varnothing$, then $\mathcal{L}_f(\mathbf{P},Z)\subseteq \mathcal{L}_f(\mathbf{P},Y)$. It is also easy to see that the Borda VCCR violates Local-$\alpha$.} Thus, one way of understanding Arrow's intention in using the example to motivate IIA is that he had in mind VIIA (cf.~\citealt[p.~180]{Bordes1991}).

\bibliographystyle{plainnat}
\bibliography{Axioms}

\end{document}